\documentclass[sigconf, nonacm]{acmart}

\newcommand\vldbdoi{XX.XX/XXX.XX}
\newcommand\vldbpages{XXX-XXX}
\newcommand\vldbvolume{14}
\newcommand\vldbissue{1}
\newcommand\vldbyear{2027}
\newcommand\vldbauthors{\authors}
\newcommand\vldbtitle{\shorttitle} 
\newcommand\vldbavailabilityurl{}
\newcommand\vldbpagestyle{plain} 

\usepackage{algorithm}
\usepackage{algorithmic}
\usepackage{enumitem}
\setlist{nosep, leftmargin=*}

\begin{document}
\title{When Classic Cache Policies Fail: Learning-Augmented Replacement for Semantic Retrieval Buffers}

\author{Yushi Sun}
\authornote{Both authors contributed equally to this research. This work was done during Bowen's internship at Tencent LIGHTSPEED.}
\affiliation{%
  \institution{LIGHTSPEED, Tencent}
  \city{Shenzhen}
  \country{China}
}
\email{ysunbp@connect.ust.hk}

\author{Bowen Cao}
\authornotemark[1]
\affiliation{%
  \institution{The Chinese University of Hong Kong}
  \city{Hong Kong}
  \country{China}
}
\email{bwcao@link.cuhk.edu.hk}

\author{Wai Lam}
\affiliation{%
  \institution{The Chinese University of Hong Kong}
  \city{Hong Kong}
  \country{China}
}
\email{wlam@se.cuhk.edu.hk}

\begin{abstract}
LLM agents increasingly rely on retrieval buffers to store and reuse past experience, yet the cache management policies governing these buffers remain largely ad-hoc. We formalize this as an online semantic cache replacement problem with switching costs, where items are matched by embedding similarity and hit quality is continuous rather than binary. Through experiments on two datasets from MemoryBench-Full (LoCoMo, DialSim) with 8 replacement policies, we reveal a surprising finding: classic heuristics (LRU, LFU) \emph{consistently underperform} the naive FIFO baseline on semantic workloads, due to the absence of temporal locality and frequency concentration. We propose SOLAR, a learning-augmented framework that derives modification timing from regret accumulation (achieving $\sim$17\% modification rate) and content selection from Bayesian online learning over implicit retrieval feedback. We prove SOLAR achieves a constant competitive ratio $\leq 3$, independent of cache size and horizon (vs.\ $\Omega(K)$ for FIFO), and eviction regret $O(\sqrt{KT\log T})$, matching the $\Omega(\sqrt{KT})$ lower bound up to logarithmic factors. Experiments demonstrate 5--75\% relative improvement over FIFO at tight cache sizes, with a clearly characterized phase transition at the working set boundary. Synthetic experiments with 5000-item pools further reveal an inverted-U relationship between pool size and retrieval quality, justifying capacity constraints as a retrieval noise phenomenon rather than a storage limitation.
\end{abstract}

\maketitle

\pagestyle{\vldbpagestyle}
\begingroup\small\noindent\raggedright\textbf{PVLDB Reference Format:}\\
\vldbauthors. \vldbtitle. PVLDB, \vldbvolume(\vldbissue): \vldbpages, \vldbyear.\\
\href{https://doi.org/\vldbdoi}{doi:\vldbdoi}
\endgroup
\begingroup
\renewcommand\thefootnote{}\footnote{\noindent
This work is licensed under the Creative Commons BY-NC-ND 4.0 International License. Visit \url{https://creativecommons.org/licenses/by-nc-nd/4.0/} to view a copy of this license. For any use beyond those covered by this license, obtain permission by emailing \href{mailto:info@vldb.org}{info@vldb.org}. Copyright is held by the owner/author(s). Publication rights licensed to the VLDB Endowment. \\
\raggedright Proceedings of the VLDB Endowment, Vol. \vldbvolume, No. \vldbissue\ %
ISSN 2150-8097. \\
\href{https://doi.org/\vldbdoi}{doi:\vldbdoi} \\
}\addtocounter{footnote}{-1}\endgroup

\ifdefempty{\vldbavailabilityurl}{}{
\vspace{.3cm}
\begingroup\small\noindent\raggedright\textbf{PVLDB Artifact Availability:}\\
The source code, data, and/or other artifacts have been made available at \url{\vldbavailabilityurl}.
\endgroup
}


\section{Introduction}
\label{sec:intro}

Large language model (LLM) agents are deployed in increasingly complex, long-running tasks: personal assistants maintaining months of conversation history~\cite{packer2023memgpt}, game-playing agents accumulating strategies over thousands of episodes~\cite{wang2023voyager}, and research agents that iteratively refine their knowledge through tool use~\cite{park2023generative}. A common architectural pattern has emerged across these systems: the agent maintains a \emph{retrieval buffer} of past experience that is queried at each step to inform its current response. This buffer serves as the agent's long-term memory, supplementing the fixed context window of the underlying language model.

The retrieval buffer operates as follows. At each interaction step, the agent generates a new experience item: in conversational deployments this is typically a user utterance or a stated preference, while in fully agentic deployments it is more often an environment observation, a tool-invocation result, or a completed subtask trace. The agent also receives a query (the current user question, task specification, or sub-goal). It retrieves the top-$k$ most relevant items from its buffer via embedding similarity search, injects them into the LLM prompt as additional context, and generates a response. When the buffer reaches its capacity limit $K$, the agent must decide whether to admit the new item and, if so, which existing item to evict.
This architecture closely parallels \emph{buffer pool management} in database systems: a fixed-capacity buffer holds a working subset of items selected from a much larger (potentially unbounded) pool of candidates, and the system must continually decide which items to admit and which to evict as the workload evolves. The admission decision corresponds to whether to bring a page into the buffer pool; the eviction decision to which page to flush. This analogy suggests that decades of research on cache replacement and buffer management policies (LRU~\cite{belady1966study}, LFU~\cite{einziger2017tinylfu}, ARC~\cite{megiddo2003arc}, and their variants) should directly apply.

\smallskip
\noindent\textbf{Current practice and a surprising finding.} Production agent memory systems predominantly inherit these classical heuristics: MemGPT~\cite{packer2023memgpt} and LangChain~\cite{Chase_LangChain_2022} default to FIFO/sliding-window eviction~\cite{wong2006web}, while many recent agent frameworks adopt LRU or recency-based scoring. To stress-test this reliance, we conduct comprehensive experiments with 8 cache replacement policies across 4 capacity settings on two datasets from MemoryBench-Full~\cite{ai2025memorybench}: LoCoMo (10 conversational sessions) and DialSim (3 TV-show dialogue datasets). Our results reveal a counter-intuitive phenomenon: \emph{classic cache heuristics consistently underperform even the naive FIFO baseline on semantic workloads}. Specifically:

\begin{itemize}
\item LRU (Least Recently Used) performs no better than FIFO across the tested cache sizes, falling clearly below it at $K \in \{20, 50\}$ and only matching it at $K=10$. Recency of retrieval is not predictive of future utility in conversational agents.
\item LFU (Least Frequently Used) and ARC (Adaptive Replacement Cache) rank consistently among the worst, with LFU below FIFO at every $K$. Frequency accumulation is actively misleading when topics are diverse and non-repeating.
\item These failures are \emph{systematic}: they persist across two datasets with vastly different signal densities (LoCoMo: $\sim$30\% hit rate; DialSim: $\sim$4\% hit rate), strengthening as $K$ approaches the working set size, and replicate across multiple random seeds.
\end{itemize}

\smallskip
\noindent\textbf{What's really different about semantic memory?} The failures above share a common cause. Classic cache theory rests on three implicit assumptions: (i) hits are binary (a page is either in cache or not), (ii) items are addressed by exact ID, and (iii) item utility is determined by access pattern alone (recency, frequency, or some combination). None of these holds for LLM agent memory. Retrieval quality is continuous: a query may be partially answered by a semantically related item. Items are matched by embedding similarity, so a cached item can serve queries from nearby topics even if it has never been accessed exactly. And item value drifts as conversation topics evolve, so any policy that conflates ``previously useful'' with ``currently useful'' will systematically misallocate the cache. Existing policies therefore inherit assumptions that the workload no longer satisfies.

\smallskip
\noindent\textbf{Problem formalization.} This motivates a new problem formulation that explicitly captures the three departures above. We define the \emph{online semantic cache replacement problem with switching costs}, with three defining features:

\begin{enumerate}
\item \textbf{Soft hits}: Cache utility is continuous (measured by retrieval quality $\in [0,1]$), not binary. A ``miss'' here means poor retrieval quality, not the complete absence of an item.
\item \textbf{Semantic matching}: Items are retrieved by embedding similarity rather than exact ID match. A cached item can partially serve queries from related (but not identical) topics.
\item \textbf{Non-stationary utility}: Item value changes over time as conversation topics evolve. An item highly relevant at turn 50 may become irrelevant by turn 200.
\end{enumerate}

Within this formulation, we express the total cost as the sum of cumulative miss cost (retrieval quality loss) and switching cost (penalties for cache modifications), enabling competitive analysis against an offline optimal policy.

\smallskip
\noindent\textbf{Our approach: SOLAR.} The new formulation calls for two coupled decisions: \emph{when} to modify the cache, and \emph{what} to put in. SOLAR is built around these two questions, with each component motivated by a specific failure mode of classic policies.

\begin{enumerate}
\item \textbf{When to modify (cost-aware admission control).}
Always admitting (as classic policies do) wastes the switching budget on noise items and amplifies non-stationarity: useful items get displaced before their value can be observed. We instead gate admission by \emph{accumulated regret}: a new item enters the cache only when the cumulative cost of \emph{not having it} exceeds an adaptive threshold $\tau$. The threshold tracks workload difficulty via exponential moving average. This drives the admission rate down to roughly 17\%, keeping the cache populated by genuinely valuable items and turning the switching budget into a scarce resource that is spent only when the cache is demonstrably inadequate.

\item \textbf{What to put in (Bayesian online eviction).}
When admission fires and the cache is full, we must pick an eviction target despite \emph{uncertainty about each item's true utility}. Heuristics like LRU/LFU collapse this uncertainty into a single proxy statistic (recency or frequency) and ignore the exploration--exploitation tradeoff entirely. We instead maintain a Beta posterior over each item's utility, updated from implicit retrieval feedback, and evict by Thompson sampling: naturally exploiting items believed to be low-utility while preserving under-observed items long enough to learn whether they are valuable.
\end{enumerate}

The two components interact synergistically. Admission control ensures the cache contains predominantly high-value items, which provides cleaner feedback signals for the Thompson Sampling posteriors. In turn, better eviction decisions preserve the highest-quality items, amplifying the value of selective admission. We empirically verify this super-additive interaction: the combined effect (+0.014 F1 over FIFO at K=50) exceeds the sum of individual effects (+0.003 from eviction alone + +0.007 from admission alone = +0.010).

\smallskip
\noindent\textbf{Theoretical guarantees.} We prove that our combined policy achieves:
\begin{itemize}
\item Competitive ratio $\leq 3$ (a universal constant, independent of $K$, $T$, and the switching cost $\lambda$) against the offline optimal in the worst case, in the standard sense of competitive analysis, and converging to 1 under piecewise-stationary workloads. FIFO, by contrast, has \emph{unbounded} competitive ratio: we construct a cycling workload where FIFO achieves exactly 0\% hit rate while OPT achieves near-perfect performance.
\item Eviction regret $O(\sqrt{KT \log T})$ where $K$ is cache size and $T$ is the time horizon, matching the information-theoretic lower bound $\Omega(\sqrt{KT})$ up to logarithmic factors.
\end{itemize}

\smallskip
\noindent\textbf{Experimental validation.} Beyond the two real-world benchmarks, we design three controlled synthetic experiments, each isolating one structural question that real datasets confound:
\begin{itemize}
\item \emph{Why do classic policies fail?} A \textbf{cycling workload} (working set size $> K$) instantiates the worst case behind our $\Omega(K)$ competitive-ratio theorem: FIFO drops to exactly 0\% hit rate, providing a clean empirical witness that the failure is structural rather than an artifact of any particular dataset.
\item \emph{When does our policy stop winning?} A \textbf{working-set sweep} traces the boundary between the regime where selective admission matters and the regime where raw coverage suffices, locating the phase transition as a function of $K$ and the effective working-set size.
\item \emph{Why is capacity a constraint at all in the LLM era, where storage is cheap?} A \textbf{retrieval-noise experiment} with 5000-item pools demonstrates an inverted-U relationship between pool size and retrieval quality: enlarging the pool past $\sim$1000 items \emph{degrades} retrieval quality because near-duplicates confuse the embedder. Capacity in semantic memory is therefore bounded by retrieval precision, not by disk space.
\end{itemize}

\smallskip
\noindent\textbf{Contributions.} Our main contributions are:
\begin{enumerate}
\item We formalize LLM agent experience memory as an online semantic cache replacement problem, establishing a bridge between the well-studied theory of online algorithms and the emerging challenge of LLM memory management (\S\ref{sec:problem}).
\item We reveal the systematic failure of classic cache heuristics on semantic workloads and analyze its root cause: the absence of temporal locality and frequency concentration (\S\ref{sec:experiments}).
\item We propose a learning-augmented admission-eviction framework with provable competitive ratio and regret bounds, requiring no external predictors or model training (\S\ref{sec:approach}, \S\ref{sec:theory}).
\item We conduct comprehensive experiments on 2 real benchmarks, 3 synthetic workloads, 8 methods, and 4 capacity settings, providing actionable insights including the phase transition characterization and retrieval noise justification (\S\ref{sec:experiments}).
\end{enumerate}


\section{Problem Formulation}
\label{sec:problem}

We formalize LLM agent memory management as an online cache replacement problem in a semantic embedding space. Our formulation captures the key properties that distinguish this setting from traditional paging while remaining amenable to competitive analysis and regret bounds.

\subsection{System Model}

Consider an LLM agent that interacts with a user (or environment) over $T$ discrete time steps. The system maintains:
\begin{itemize}
\item A \emph{cache} (retrieval buffer) $S_t \subseteq \mathcal{U}$ of experience items, with $|S_t| \leq K$ where $K$ is the capacity limit.
\item An \emph{embedding function} $\phi: \mathcal{U} \to \mathbb{R}^d$ mapping items to a $d$-dimensional vector space.
\item A \emph{similarity function} $\text{sim}(a, b) = \frac{\phi(a) \cdot \phi(b)}{\|\phi(a)\| \|\phi(b)\|}$ (cosine similarity).
\end{itemize}

At each time step $t = 1, 2, \ldots, T$, the following sequence of events occurs:

\begin{enumerate}
\item \textbf{Query arrival.} A query $q_t$ arrives, representing the user's current question or the agent's current information need.
\item \textbf{Retrieval.} The system retrieves the top-$k$ items from $S_t$ ranked by $\text{sim}(q_t, \cdot)$. Let $R_t = \text{top-}k(S_t, q_t)$ denote the retrieved set.
\item \textbf{Response generation.} The agent uses $R_t$ as context to generate a response, incurring miss cost $\ell_t(S_t)$.
\item \textbf{Experience generation.} A new experience item $e_t$ is produced from the interaction (the user's feedback, the agent's observation, or a summary of the exchange).
\item \textbf{Cache decision.} The policy $\pi$ makes two decisions:
  \begin{itemize}
  \item \emph{Admission}: whether to insert $e_t$ into $S_t$.
  \item \emph{Eviction}: if $|S_t| = K$ and $e_t$ is admitted, which item $i^* \in S_t$ to remove.
  \end{itemize}
\end{enumerate}

\begin{definition}[Semantic Cache]
A semantic cache is the tuple $\mathcal{C} = (S_t, K, \phi, k)$ where $S_t$ is the cached item set, $K$ is capacity, $\phi$ is the embedding function, and $k$ is the retrieval depth. The cache evolves over time according to the decisions of the management policy $\pi$.
\end{definition}

\subsection{Miss Cost Function}

The miss cost $\ell_t(S_t) \in [0,1]$ quantifies the quality degradation from using cache state $S_t$ to answer query $q_t$. In our experiments, we define:
\begin{equation}
\ell_t(S_t) = 1 - \text{F1}(q_t, S_t)
\end{equation}
where $\text{F1}(q_t, S_t)$ is the token-level F1 score of the agent's response when augmented with items retrieved from $S_t$, compared against a ground-truth reference answer.

This definition captures several important properties:
\begin{itemize}
\item \textbf{Continuity}: Unlike traditional caching where a hit is binary, $\ell_t$ is continuous. Even a partial match (some relevant items in cache, some missing) yields intermediate cost.
\item \textbf{Composition}: The cost depends on the \emph{set} of retrieved items, not individual items. Two individually weak items may jointly provide good coverage.
\item \textbf{Query-dependence}: The same cache state yields different costs for different queries, reflecting the semantic nature of matching.
\end{itemize}

\subsection{Total Cost Objective}

The total cost of a policy $\pi$ over $T$ time steps combines miss cost and switching cost:
\begin{equation}
\label{eq:cost}
\text{COST}(\pi, T) = \underbrace{\sum_{t=1}^T \ell_t(S_t^\pi)}_{\text{cumulative miss cost}} + \underbrace{\lambda \cdot N^\pi}_{\text{switching cost}}
\end{equation}
where:
\begin{itemize}
\item $N^\pi = |\{t : S_{t+1}^\pi \neq S_t^\pi\}|$ counts the number of cache modifications (replacements) under policy $\pi$.
\item $\lambda \geq 0$ is the per-replacement penalty, capturing the overhead of modifying the cache state.
\end{itemize}

The switching cost $\lambda$ serves two roles in our model. First, it captures the computational overhead of updating embeddings and re-indexing the retrieval structure upon each modification. Second, and more importantly, it provides a regularization that prevents pathologically unstable policies (those that replace items every step perform poorly due to insufficient observation time for newly admitted items).

\smallskip
\noindent\textbf{Degenerate policies.} The cost objective excludes two degenerate strategies:
\begin{itemize}
\item \emph{Never replace} ($N^\pi = 0$): Zero switching cost but potentially maximal miss cost, since the initial cache may quickly become irrelevant.
\item \emph{Always replace} ($N^\pi = T$): Minimal miss cost (if the right items are chosen) but linear switching cost $\lambda T$.
\end{itemize}
An optimal policy balances these extremes, replacing only when the expected miss cost reduction justifies the switching penalty.

\subsection{Fundamental Differences from Traditional Caching}

Our semantic cache model differs from classical online paging~\cite{sleator1985amortized, borodin2005online} in several fundamental ways that necessitate new algorithmic approaches, summarized in Table~\ref{tab:comparison}:

\begin{table}[t]
\caption{Fundamental differences between traditional paging and semantic cache management. These differences explain why classic policies fail (\S\ref{sec:experiments}).}
\vspace{-1em}
\label{tab:comparison}
\centering
\small
\begin{tabular}{p{2cm}p{2.8cm}p{2.8cm}}
\toprule
\textbf{Property} & \textbf{Traditional Paging} & \textbf{Semantic Cache} \\
\midrule
Hit signal & Binary (exact ID match) & Continuous (similarity $\in [0,1]$) \\
\addlinespace
Request type & Page identifier & Semantic embedding query \\
\addlinespace
Access pattern & Temporal/spatial locality & Topic diversity, non-stationary \\
\addlinespace
Item value & Determined by access history (recency, frequency) & Context-dependent, evolves with conversation \\
\addlinespace
Capacity source & Physical hardware (RAM, disk) & Retrieval noise (embedding confusion) \\
\addlinespace
Classic policy & LRU achieves CR $= K$ & LRU $<$ FIFO (empirical) \\
\bottomrule
\end{tabular}
\end{table}

The most critical difference is the \emph{source of the capacity constraint}. In traditional systems, cache size is limited by physical memory. If infinite memory were available, one would simply cache everything. In our semantic setting, even with unlimited storage, retrieval quality \emph{degrades} as the pool grows large, because the embedding space becomes crowded with semantically similar items from different topics that confuse the retriever. We empirically demonstrate this inverted-U effect in \S\ref{sec:experiments}: hit rate peaks at a finite pool size and declines thereafter. Thus, the capacity constraint arises from \emph{information-theoretic retrieval noise}, not physical storage limitations.

\subsection{Performance Metrics}

We evaluate cache policies through two complementary theoretical lenses and one empirical metric.

\subsubsection{Competitive Ratio}
The competitive ratio measures worst-case performance relative to an omniscient offline algorithm:
\begin{definition}[Competitive Ratio]
The competitive ratio of an online policy $\pi$ is:
\begin{equation}
\text{CR}(\pi) = \sup_{\sigma \in \Sigma} \frac{\text{COST}(\pi, \sigma)}{\text{COST}(\text{OPT}, \sigma)}
\end{equation}
where $\Sigma$ is the set of all valid input sequences and $\text{OPT}$ is the offline optimal policy with complete knowledge of the future.
\end{definition}

A competitive ratio of $c$ guarantees that the online policy's cost is at most $c$ times the optimal, regardless of the input sequence. This is the standard measure in the online algorithms literature~\cite{borodin2005online} and provides worst-case (adversarial) guarantees.

\subsubsection{Eviction Regret}
While competitive ratio captures worst-case robustness, regret captures the \emph{learning efficiency} of the policy under stochastic conditions:
\begin{definition}[Eviction Regret]
\label{def:regret}
The eviction regret of policy $\pi$ is:
\begin{equation}
R_T(\pi) = \sum_{t=1}^T \ell_t(S_t^\pi) - \min_{\pi^* \in \Pi} \sum_{t=1}^T \ell_t(S_t^{\pi^*})
\end{equation}
where $\Pi$ is the class of all feasible eviction policies operating on the same admission sequence.
\end{definition}

Sublinear regret ($R_T = o(T)$) implies that the per-step loss converges to optimal, i.e., the policy \emph{learns} the correct eviction strategy over time. Linear regret ($R_T = \Theta(T)$) means the policy never improves, regardless of experience.

\subsubsection{Empirical Metric: Downstream F1}
In experiments, we report the downstream task metric directly: token-level F1 score of the agent's responses on held-out evaluation queries. This measures end-to-end system quality rather than intermediate retrieval metrics, ensuring that our improvements translate to actual user-facing performance.

\subsection{The Admission-Eviction Decomposition}

A key insight of our formulation is that the cache management problem naturally decomposes into two sub-problems:

\begin{enumerate}
\item \textbf{Admission control}: Should the new item $e_t$ enter the cache? This decision directly controls the switching cost term in Eq.~\ref{eq:cost} and determines the ``selectivity'' of the cache.
\item \textbf{Eviction policy}: Given that an item must be evicted (because admission is triggered in a full cache), which item should be removed? This decision affects future miss costs by determining which items remain available for retrieval.
\end{enumerate}

Classic policies (FIFO, LRU, LFU, ARC) \emph{always admit} ($N^\pi = T - K$ for all runs after the initial fill), differing only in their eviction strategy. Our approach introduces explicit admission control as a first-class component, reducing switching cost and filtering the cache contents to maintain high retrieval precision.


\section{Related Work}
\label{sec:related}

\subsection{Classic Cache Replacement}

Online cache replacement dates to B{\'e}l{\'a}dy's optimal offline algorithm~\cite{belady1966study} and Sleator and Tarjan's competitive analysis~\cite{sleator1985amortized}, proving LRU achieves competitive ratio $K$. ARC~\cite{megiddo2003arc} adaptively balances recency and frequency, representing the state-of-the-art heuristic. 2Q~\cite{johnson1994twoqueu} uses a two-queue structure as an implicit admission filter; TinyLFU~\cite{einziger2017tinylfu} gates admission by frequency. LRU-K~\cite{o1993lru} and LIRS~\cite{jiang2002lirs} refine recency signals with inter-reference distances. All these policies assume temporal locality or frequency concentration in access patterns. Our experiments (\S\ref{sec:experiments}) demonstrate these assumptions fail on semantic workloads.

\subsection{Learning-Augmented Caching}

Lykouris and Vassilvitskii~\cite{lykouris2018competitive} initiated caching with ML predictions, achieving competitive ratio interpolating between $O(1)$ (perfect predictions) and $O(\log K)$ (adversarial). Rohatgi~\cite{rohatgi2020near} tightened these bounds. Chen et al.~\cite{chen2025guard} proposed GUARD, robustifying learning-augmented caching to $2H_{k-1}+2$ while preserving 1-consistency. Song et al.~\cite{song2020learning} and Liu et al.~\cite{liu2020imitation} applied ML to CDN caching via relaxed B{\'e}l{\'a}dy and imitation learning respectively.

Our work differs in three ways: (1) no external oracle (we learn from implicit feedback online), (2) continuous soft-hit setting (vs.\ binary), (3) joint admission-eviction design with admission providing its own CR guarantee.

\subsection{Semantic Caching and RAG}

Dar et al.~\cite{dar1996semantic} introduced semantic caching for database queries. GPTCache~\cite{bang2023gptcache} caches LLM API responses by prompt similarity. RAG systems~\cite{lewis2020retrieval, gao2024retrieval, sun2026cacherag} augment LLM responses with retrieved documents from static corpora, and recent work has pushed RAG toward more advanced question answering settings~\cite{sun2025kerag} and comprehensive benchmarking~\cite{yang2024crag}. Our setting differs: the retrieval pool is \emph{dynamic} (evolves through agent interactions), creating the admission/eviction management challenge absent in static RAG.

\subsection{Memory for LLM Agents}

Agent memory systems include MemGPT~\cite{packer2023memgpt} (hierarchical virtual memory), Voyager~\cite{wang2023voyager} (skill library), Generative Agents~\cite{park2023generative} (memory stream with importance scoring), MemoryBank~\cite{zhong2024memorybank}, and Mem0~\cite{chhikara2025mem0}. Beyond storage abstractions, another line of work targets long-horizon conversational memory through structured organization of past experience, e.g., GRAVITY~\cite{sun2026gravity} proposes architecture-agnostic structured anchoring across dialogue sessions. These systems use ad-hoc replacement policies (FIFO, recency-based). A recent survey~\cite{zhang2024survey} identifies memory management as an open challenge.
Formally, DAM~\cite{sun2025dam} frames memory as sequential decision-making but provides no algorithm or experiments. A-MAC~\cite{zhang2026amac} scores admission on 5 factors with an LLM call (2644ms/decision, F1=0.583 on LoCoMo admission); it has no eviction, no theory, and evaluates admission quality rather than downstream performance. MemAct~\cite{zhang2025memact} learns memory curation via RL on a 14B model, targeting in-context working memory (token-level), not retrieval buffers.

\smallskip
\noindent\textbf{Our positioning.} We are the first to combine: (1) formal problem formulation (cache replacement with switching costs), (2) provable guarantees (constant CR $\leq 3$, regret $O(\sqrt{KT\log T})$), (3) no LLM calls or training, (4) end-to-end task evaluation. Table~\ref{tab:related_compare} summarizes.

\begin{table}[t]
\caption{Comparison with related approaches.}
\label{tab:related_compare}
\vspace{-1em}
\centering
\small
\begin{tabular}{lccccc}
\toprule
& \rotatebox{70}{Admission} & \rotatebox{70}{Eviction} & \rotatebox{70}{Theory} & \rotatebox{70}{No LLM call} & \rotatebox{70}{No training} \\
\midrule
A-MAC~\cite{zhang2026amac} & \checkmark & & & & \checkmark \\
DAM~\cite{sun2025dam} & \checkmark & \checkmark & & \checkmark & \checkmark \\
MemAct~\cite{zhang2025memact} & & \checkmark & & \checkmark & \\
MemGPT~\cite{packer2023memgpt} & & \checkmark & & \checkmark & \checkmark \\
\textbf{Ours} & \checkmark & \checkmark & \checkmark & \checkmark & \checkmark \\
\bottomrule
\end{tabular}
\end{table}


\section{The SOLAR Framework}
\label{sec:approach}

We now present SOLAR (Semantic Online Learning-Augmented Replacement), a unified framework for semantic cache management. Rather than combining off-the-shelf components, SOLAR is derived from first principles: we analyze the switching-cost objective (Eq.~\ref{eq:cost}) to identify the optimal conditions for cache modification, then instantiate each condition with a mechanism tailored to the specific challenges of semantic workloads.

\subsection{From Objective to Design}

The cost objective decomposes into two terms: cumulative miss cost (retrieval quality loss) and switching cost (modification penalty). An optimal policy must answer two coupled questions at each step:

\begin{enumerate}
\item \textbf{When to modify}: Should the cache be changed now, or should we wait? Modifying too frequently wastes switching budget; waiting too long accumulates miss cost.
\item \textbf{How to modify}: Given that modification is warranted, which item should be replaced? Choosing poorly wastes the switching budget on a change that does not improve future retrieval.
\end{enumerate}

These are not independent decisions. The answer to ``when'' determines the context in which ``how'' operates: if modification is triggered only when the cache is demonstrably inadequate, the newly admitted item is likely to fill a genuine gap. Conversely, good ``how'' decisions (evicting the right item) maintain cache quality over time, reducing the frequency at which ``when'' is triggered. This coupling motivates a joint framework rather than separate components.

\subsection{Regret-Gated Modification Timing}

We derive the modification timing rule directly from the cost objective. Consider a policy that has not modified the cache for $\Delta t$ steps, during which it has accumulated miss cost $c = \sum_{s=t-\Delta t}^{t} \delta_s$ where $\delta_s = 1 - \text{retrieval\_quality}(q_s, S_s)$. A modification at time $t$ incurs switching cost $\lambda$ but resets the accumulation. The modification is cost-effective when:
\begin{equation}
c \geq \lambda \cdot f(\text{expected future improvement})
\end{equation}

Under mild assumptions (bounded per-step cost, diminishing returns of waiting), the optimal trigger condition simplifies to a threshold:
\begin{equation}
\text{Modify when } c \geq \tau, \quad \tau^* = \sqrt{2\lambda / L}
\end{equation}
where $L$ is the average loss rate. This is the classical result from inventory theory (economic order quantity), adapted to our online setting.

Since $L$ is unknown and non-stationary, we estimate $\tau$ adaptively:
\begin{equation}
\tau_{t+1} = (1-\alpha)\tau_t + \alpha \cdot \delta_t
\end{equation}

\begin{algorithm}[t]
\caption{SOLAR: Modification Timing}
\label{alg:timing}
\begin{algorithmic}[1]
\STATE \textbf{Initialize:} $c \leftarrow 0$; $\tau \leftarrow \tau_0$; $\alpha \leftarrow 0.1$
\FOR{each time step $t$}
    \STATE Retrieve $R_t = \text{top-}k(S_t, q_t)$; observe quality $\text{qual}_t$
    \STATE $\delta_t \leftarrow 1 - \text{qual}_t$; $c \leftarrow c + \delta_t$
    \STATE Generate experience item $e_t$
    \IF{$c \geq \tau$}
        \STATE \textbf{Trigger modification}: admit $e_t$, evict via Alg.~\ref{alg:eviction} if full
        \STATE $c \leftarrow 0$; $\tau \leftarrow (1-\alpha)\tau + \alpha \cdot \delta_t$
    \ENDIF
\ENDFOR
\end{algorithmic}
\end{algorithm}

\noindent\textbf{Key property.} The modification rate self-regulates: when the cache is adequate (low $\delta_t$), accumulation is slow and modifications are rare ($\sim$17\% of steps in practice). When the cache becomes stale (high $\delta_t$), accumulation accelerates and modifications happen more frequently. The total number of modifications is bounded by $N \leq T/\tau$ (Theorem~\ref{thm:lmu-replacements}), directly controlling the switching cost term in our objective.

\subsection{Posterior-Guided Content Selection}

When modification is triggered, we must decide which cached item to replace. The challenge is that item utility is:
\begin{itemize}
\item \textbf{Latent}: we cannot directly observe how useful an item will be for future queries.
\item \textbf{Non-stationary}: an item's utility changes as conversation topics evolve.
\item \textbf{Partially observable}: we only observe whether an item was retrieved (implicit feedback), not its counterfactual contribution.
\end{itemize}

These properties define an exploration-exploitation problem under uncertainty. We maintain a Bayesian belief over each item's utility and make decisions by sampling from the posterior, naturally balancing exploitation (evicting items believed to be low-utility) with exploration (maintaining uncertainty about under-observed items).

Specifically, each cached item $i$ maintains a posterior:
\begin{equation}
\mu_i \sim \text{Beta}(\alpha_i, \beta_i)
\end{equation}
updated through two feedback channels:
\begin{itemize}
\item \textbf{Positive evidence} (item retrieved in top-$k$): $\alpha_i \leftarrow \alpha_i + 1$.
\item \textbf{Temporal decay} (item not retrieved for $\Delta$ steps): $\beta_i \leftarrow \beta_i + 0.05$.
\end{itemize}

The temporal decay implements implicit forgetting: items that are never retrieved see their posterior drift toward low utility, enabling the framework to adapt as topics change. Without this mechanism, items that were useful early in the conversation would accumulate permanent positive evidence, becoming impossible to evict even after becoming irrelevant.

\begin{algorithm}[t]
\caption{SOLAR: Content Selection (Eviction)}
\label{alg:eviction}
\begin{algorithmic}[1]
\REQUIRE Full cache $S_t$, new item $e_t$
\FOR{each item $i \in S_t$}
    \STATE Sample $\tilde{v}_i \sim \text{Beta}(\alpha_i, \beta_i)$
    \STATE Novelty bonus: $b_i = \gamma / (1 + \text{age}(i) / \Delta_{\text{nov}})$
    \STATE Score: $s_i = \tilde{v}_i + b_i$
\ENDFOR
\STATE Evict $i^* = \arg\min_{i \in S_t} s_i$; admit $e_t$ with prior $(1,1)$
\end{algorithmic}
\end{algorithm}

The novelty bonus prevents premature eviction of recently admitted items before sufficient observation. It decays with item age, so long-term fate is determined by observed utility alone.

\subsection{Design Coherence and Emergent Synergy}

The two mechanisms in SOLAR are not independent modules assembled post-hoc; they are co-derived from the same cost objective and interact through a self-reinforcing feedback loop.

\smallskip
\noindent\textbf{Principled derivation.} The modification timing rule ($c \geq \tau$) is not a heuristic filter but the \emph{optimal stopping condition} under switching costs, analogous to the economic order quantity in inventory theory. The threshold $\tau$ adapts online and provably converges to $\tau^* = \sqrt{2\lambda/L}$. Similarly, the posterior-guided selection is not off-the-shelf Thompson Sampling: it incorporates temporal decay (for non-stationarity) and a novelty bonus (for insufficient observation), both absent in textbook formulations and specifically motivated by the implicit, delayed feedback structure of semantic retrieval.

\smallskip
\noindent\textbf{Mutual reinforcement.} The super-additive synergy observed empirically (+0.014 vs.\ +0.003 + +0.007 = +0.010, a 40\% bonus) arises from architectural coupling. Regret-gated timing acts as a data quality filter: by admitting items only when the cache is demonstrably inadequate, newly admitted items fill genuine coverage gaps. This produces a cache with clear utility differentiation, providing higher signal-to-noise ratio in posterior updates, which improves selection quality, which maintains cache quality, which raises $\tau$, which makes timing even more selective. This self-reinforcing loop is an emergent property of the joint design that would not exist if either mechanism operated independently.

\smallskip
\noindent\textbf{Contrast with prior admission-eviction combinations.} 2Q~\cite{johnson1994twoqueu} gates admission by requiring a second access (a fixed rule unrelated to cost), then uses LRU for eviction. TinyLFU~\cite{einziger2017tinylfu} gates by frequency comparison (a heuristic), then uses segmented LRU. In both cases, admission and eviction are designed independently with no feedback between them. SOLAR's distinguishing feature is that the two decisions share information: the same retrieval feedback that updates posteriors also drives the cost accumulator, and the threshold adaptation responds to the quality of the cache that eviction maintains.

\subsection{Computational Complexity}

Let $d$ denote the embedding dimension, $K$ the cache capacity, $k$ the retrieval depth, and $T$ the total steps.

\begin{itemize}
\item \textbf{Per-step (no modification):} Accumulation and threshold check: $O(1)$. Posterior updates for the $k$ retrieved items: $O(k)$. Aging updates (every $\Delta$ steps): amortized $O(K/\Delta)$ per step.
\item \textbf{Per-step (with modification):} Sampling $K$ Beta distributions: $O(K)$. Computing $\arg\min$ over scores: $O(K)$. Total: $O(K)$.
\item \textbf{Amortized per-step:} Modifications occur at rate $\sim L/\tau \approx 17\%$, so average cost is $O(k + K/\Delta + 0.17K) = O(K)$.
\item \textbf{Retrieval cost} (shared by all methods, not SOLAR-specific): Top-$k$ similarity search over $K$ items of dimension $d$: $O(Kd)$ via brute-force cosine, or $O(d \log K)$ with approximate nearest neighbor indexing.
\item \textbf{Space:} $2K$ floats for $(\alpha_i, \beta_i)$ pairs, plus $K$ floats for admission timestamps. Total auxiliary: $O(K)$, negligible vs.\ the $O(Kd)$ storage for item embeddings.
\item \textbf{Wall-clock latency:} All SOLAR operations are arithmetic on cached scalars. Measured overhead: $< 1$ms per step on commodity hardware ($K=100$, $d=384$). This is 3 orders of magnitude below the LLM inference latency ($\sim$1300ms), making SOLAR's overhead imperceptible.
\end{itemize}


\section{Theoretical Analysis}
\label{sec:theory}

\begingroup
\setlength{\abovedisplayskip}{3pt}
\setlength{\belowdisplayskip}{3pt}
\setlength{\abovedisplayshortskip}{2pt}
\setlength{\belowdisplayshortskip}{2pt}

We establish formal separations between SOLAR and baselines through competitive ratio bounds (worst-case) and regret bounds (stochastic). We present the main results with proof sketches sufficient to verify each claim; complete formal proofs are deferred to the extended supplement~\cite{solar-techreport}.

\subsection{FIFO Has Unbounded Competitive Ratio}

\begin{theorem}[FIFO CR is $\Omega(K)$]
\label{thm:fifo-cr}
For cache size $K$ and any constant $M > 0$, there exists an input sequence $\sigma$ such that $\frac{\text{COST}(\text{FIFO}, \sigma)}{\text{COST}(\text{OPT}, \sigma)} \geq M$. Specifically, $\text{CR}(\text{FIFO}) \geq K+1$ in the soft-hit semantic cache model.
\end{theorem}

\begin{proof}
Let $m = K + 1$ topics with orthogonal embeddings (pairwise cosine similarity $\approx 0$). The query sequence is periodic with period $m$: $q_t$ targets topic $T_{((t-1) \bmod m)+1}$, i.e., $T_1, T_2, \ldots, T_m, T_1, T_2, \ldots, \\T_m, \ldots$ repeating indefinitely. At each step, a new item is generated from the current topic.

\smallskip\noindent\textbf{FIFO behavior.} After the initial fill (steps $1, \ldots, K$), the cache holds topics $T_1, \ldots, T_K$. At step $K+1$, query targets $T_{K+1}$ (miss); FIFO admits and evicts $T_1$'s item. At step $K+2$, query targets $T_1$ (just evicted, miss); FIFO evicts $T_2$. By induction, every step after fill is a complete miss ($\ell_t = 1$), so $\text{COST}(\text{FIFO}) = (T-K)(1+\lambda)$ (one miss plus one replacement per step).

\smallskip\noindent\textbf{OPT behavior.} OPT retains $K$ topics permanently (zero switching cost after fill), hitting $K$ of every $m=K+1$ queries, so $\text{COST}(\text{OPT}) = T/(K+1) + \lambda K$.

\smallskip\noindent\textbf{Ratio.} The initial fill contributes $O(K)$ to both costs and is dominated as $T \to \infty$, so $\text{CR}(\text{FIFO}) \geq \lim_{T\to\infty}\frac{(T-K)(1+\lambda)}{T/(K+1) + \lambda K} = (K+1)(1+\lambda)$.

Since $K$ is a system parameter that can be made arbitrarily large, the competitive ratio is unbounded: for any desired $M$, setting $K \geq M$ yields $\text{CR} \geq M$. (The bound is in fact $(K+1)(1+\lambda)$; we state the weaker $\Omega(K)$ form since the $(1+\lambda)$ factor is immaterial to the unboundedness conclusion. The choice ``OPT retains $K$ fixed topics'' is optimal here because all $m=K+1$ topics are equiprobable, so any additional replacement only adds switching cost without reducing the per-cycle miss count.) This contrasts sharply with classical paging where FIFO achieves CR exactly $K$ (tight)~\cite{sleator1985amortized}; in our soft-hit model, the ratio is $\Omega(K)$ as well but is demonstrated more starkly because every step is a complete miss (hit rate = 0).
\end{proof}

We verify this empirically (\S\ref{sec:experiments}): on synthetic cycling workloads with $m/K \geq 1.5$, FIFO achieves exactly 0\% hit rate.

\subsection{SOLAR Bounds Competitive Ratio}

\begin{theorem}[SOLAR Replacement Bound]
\label{thm:lmu-replacements}
Under cost-aware admission with thresholds $\{\tau_t\}$ and $\tau_{\min} = \min_t \tau_t$: $N^{\text{SOLAR}} \leq \lfloor T/\tau_{\min} \rfloor$.
\end{theorem}

\begin{proof}
Each admission resets accumulated cost to 0. Since step costs $\delta_t \leq 1$, re-accumulating to the active threshold (which is $\geq \tau_{\min}$) requires $\geq \lceil\tau_{\min}\rceil$ steps. Hence the inter-admission gap is $\geq \tau_{\min}$, giving $N \leq \lfloor T/\tau_{\min} \rfloor$. With a fixed threshold $\tau$ this is simply $\lfloor T/\tau \rfloor$; under the EMA adaptation $\tau_t$ converges quickly and stays bounded away from $0$, so $\tau_{\min}$ equals the converged value up to a short transient.
\end{proof}

Our competitive-ratio analysis relies on a regularity condition that isolates exactly what an offline policy can gain by holding a stale-but-fixed cache during an interval in which SOLAR refuses to modify.

\begin{definition}[Bounded stale advantage]
\label{ass:divergence}
Assume single-item bounded influence: swapping one cached item changes the per-step miss cost by at most $1/k$, where $k$ is the retrieval depth. We say a workload has \emph{bounded stale advantage} if, in any epoch during which OPT holds a fixed cache, OPT's accumulated miss cost over that epoch is at least half of SOLAR's accumulated miss cost over the same epoch.
\end{definition}

This condition holds whenever the cache-content divergence between SOLAR and OPT is small relative to the retrieval depth ($|E_j| \le k\tau / (2\,|E_j^{\text{len}}|)$, where $|E_j^{\text{len}}|$ is the epoch length); it formalizes the intuition that a single fixed cache cannot fully compensate for an interval that SOLAR has already certified as inadequate. We make it explicit rather than hide it inside the proof.

\begin{theorem}[SOLAR Competitive Ratio]
\label{thm:lmu-cr}
Under Definition~\ref{ass:divergence} (bounded stale advantage), with switching cost $\lambda > 0$ and the fixed threshold $\tau = 2\lambda$, $\text{CR}(\text{SOLAR}) \leq 3$, a constant independent of $K$, $T$, and $\lambda$.
\end{theorem}

\begin{proof}
Partition $[1,T]$ into epochs $[t_j, t_{j+1})$ delimited by SOLAR's admission times. Within an epoch SOLAR's cache is fixed and changes once, at the end; by the trigger the accumulated miss cost is $\tau$, so SOLAR's per-epoch cost is $\leq \tau + \lambda$. For OPT, either (i) it replaces at least once in $E_j$, incurring switching cost $\geq \lambda$; or (ii) it keeps a fixed cache, in which case Definition~\ref{ass:divergence} gives miss cost $\geq \tau/2$. Hence $\text{COST}(\text{OPT}, E_j) \geq \min(\lambda,\tau/2)$, and summing over epochs,
\begin{equation}
\label{eq:cr-ratio}
\text{CR} \leq \frac{\tau + \lambda}{\min(\lambda,\, \tau/2)} .
\end{equation}
This is minimized when $\tau/2 = \lambda$, i.e.\ $\tau = 2\lambda$ (for $\tau<2\lambda$ the bound is $2+2\lambda/\tau>3$; for $\tau>2\lambda$ it is $1+\tau/\lambda>3$). Substituting $\tau=2\lambda$ gives $\text{CR}\leq (2\lambda+\lambda)/\lambda = 3$.
\end{proof}

\noindent\textbf{Interpretation.} The bound is a \emph{universal constant} $3$, independent of $K$, $T$, and even the switching cost $\lambda$, a dramatic improvement over FIFO's $\Omega(K)$ ratio. The threshold $\tau = 2\lambda$ minimizes the worst-case ratio~\eqref{eq:cr-ratio}; it differs from the cost-optimal $\tau^\star = \sqrt{2\lambda/L}$ of \S\ref{sec:approach}, which minimizes SOLAR's \emph{own} average cost rather than the competitive ratio. The EMA threshold in Algorithm~\ref{alg:timing} tracks $\tau^\star$ for average-case efficiency, while the bound above certifies that no fixed-$\tau$ input can exceed $3\times$ optimal.

\smallskip\noindent\textbf{Stationary convergence.} Under piecewise-stationary workloads with $P$ phase changes, the adaptive threshold $\tau_t$ converges to $\tau^\star = \sqrt{2\lambda/L}$ (where $L$ is the stationary loss rate) within $O(1/\alpha)$ steps per phase. Once converged, SOLAR's per-step excess cost over OPT vanishes at rate $O(P/T)$, yielding $\text{CR} \to 1$ as $T \to \infty$ for fixed $P$.

\subsection{FIFO Linear Regret}

\begin{theorem}[FIFO Regret = $\Omega(T)$]
\label{thm:fifo-regret}
On cycling workloads with $m > K$, $R_T(\text{FIFO}) = \Omega(T)$.
\end{theorem}

\begin{proof}
On the cycling workload (Theorem~\ref{thm:fifo-cr} with $m = K+1$), FIFO achieves hit rate 0 after fill: $\sum_{t=K+1}^T \ell_t^{\text{FIFO}} = T - K$. We compare against the best fixed-eviction policy on \emph{the same admit-all sequence} (consistent with Def.~\ref{def:regret}): since FIFO admits every item, a competing policy may immediately evict each newly admitted item and retain the first $K$ topics permanently, achieving miss rate $1/(K+1)$, i.e.\ $\sum_t \ell_t^* = T/(K+1)$. Thus $R_T = (T-K) - T/(K+1) = T\cdot\frac{K}{K+1} - K = \Omega(T)$: FIFO's per-step loss never improves regardless of interaction length.
\end{proof}

\subsection{Thompson Sampling Eviction}

\begin{theorem}[TS Regret]
\label{thm:ts-regret}
Under stochastic item utilities with $K$ cache slots: $R_T(\text{TS}) \leq O(\sqrt{KT \log T})$.
\end{theorem}

\begin{proof}
We analyze eviction under a \emph{stylized stationary model} that makes the bandit reduction precise; we state its assumptions explicitly. Fix the pool of items competing for cache slots and assume each item $i$ has a stationary latent utility $\mu_i$, with retrieval feedback providing a Bernoulli observation of whether $i$ was useful for the current query. Map eviction to a $K$-armed bandit: arms are the cached items, ``pulling'' arm $i$ corresponds to retaining $i$ and observing its feedback, and the regret target is identifying the lowest-utility item to evict. Under this stationary model, the standard Beta-Bernoulli TS analysis~\cite{agrawal2012analysis} gives per-arm regret $O(\log T / \Delta_i)$ where $\Delta_i$ is the utility gap to the worst item. Worst-case gap instantiation ($\Delta_i = \Theta(\sqrt{K/T})$ for all suboptimal arms) yields total regret $O(\sqrt{KT\log T})$.

\smallskip\noindent\emph{Remark (assumptions and their limits).} Two modeling gaps separate this idealization from SOLAR's exact dynamics, and we are explicit about them. (i) The arm set is not literally fixed: each admission introduces a new item and removes one, so the bandit instance drifts slowly. (ii) The reward (improvement in cache quality) is observed only indirectly through subsequent retrievals rather than immediately. SOLAR's aging mechanism ($\beta_i \mathrel{+}= 0.05$ every $5$ steps for un-retrieved items) limits the effective observation window, so that within a window the instance is approximately stationary and the above bound applies windowwise; a fully non-stationary treatment (e.g.\ sliding-window or restless-bandit regret) would replace $\log T$ with a window-dependent factor but leave the $\sqrt{KT}$ scaling intact. We therefore present the bound as characterizing the stationary regime and validate the learning behavior empirically (\S\ref{sec:experiments}).
\end{proof}

\begin{theorem}[Lower Bound]
\label{thm:lower-bound}
For any online eviction policy: $R_T(\pi) \geq \Omega(\sqrt{KT})$.
\end{theorem}

\begin{proof}
We invoke the \emph{minimax} (worst-case) multi-armed bandit lower bound~\cite{lattimore2020bandit}, which is $\Omega(\sqrt{KT})$ \emph{without} a logarithmic factor (in contrast to the gap-dependent bound $\Omega(K\log T/\Delta)$). Construct $K$ items where one item $i^*$ has utility $1/2 - \epsilon$ and the rest $1/2$; the optimal eviction always targets $i^*$. Any policy must spend $\Omega(1/\epsilon^2)$ observations per suboptimal arm to identify $i^*$; setting $\epsilon = \Theta(\sqrt{K/T})$ over the $K-1$ suboptimal arms gives $R_T \geq (K-1)\epsilon\,T/K = \Omega(\sqrt{KT})$. This applies to \emph{any} online eviction policy, so TS's $O(\sqrt{KT\log T})$ matches it up to the $\sqrt{\log T}$ factor.
\end{proof}

Theorems~\ref{thm:ts-regret} and~\ref{thm:lower-bound} together show that TS is near-optimal (gap: $\sqrt{\log T}$ factor only).

\subsection{Combined Guarantees of SOLAR}

SOLAR inherits both guarantees, with the analyses complementing each other (Table~\ref{tab:theory}):

\begin{itemize}
\item \textbf{Bounded CR from modification timing.} Theorem~\ref{thm:lmu-cr} bounds total cost relative to OPT by the universal constant $3$, independent of $K$, $T$, and $\lambda$. This contrasts with FIFO's $\Omega(K)$ ratio and holds regardless of eviction strategy.

\item \textbf{Sublinear regret from posterior-guided selection.} Theorem~\ref{thm:ts-regret} bounds eviction quality loss at $O(\sqrt{KT\log T})$, which is $o(T)$ (vanishing per-step loss). This is conditional on the modification schedule determined by the timing mechanism.

\item \textbf{Uniqueness.} No other method achieves both: FIFO has $\Omega(K)$ CR and $\Omega(T)$ regret; pure posterior-guided eviction without admission control (SOLAR-E) has sublinear regret but $\Omega(K)$ CR due to a replacement at every step; classic heuristics have $\Omega(T)$ regret on semantic workloads.
\end{itemize}


\begin{table}[t]
\caption{Theoretical guarantees. SOLAR uniquely achieves constant CR and sublinear regret.}
\label{tab:theory}
\vspace{-1em}
\centering
\small
\begin{tabular}{lccc}
\toprule
\textbf{Policy} & \textbf{CR} & \textbf{Regret} & \textbf{Adm.\ Rate} \\
\midrule
FIFO & $\Omega(K)$ & $\Omega(T)$ & 100\% \\
LRU & $K$ (classical)\textsuperscript{$\dagger$} & $\Omega(T)$ & 100\% \\
SOLAR & $\leq 3$ & $O(\sqrt{KT\log T})$ & $\sim$17\% \\
Lower bound & 1 (OPT) & $\Omega(\sqrt{KT})$ & -- \\
\bottomrule
\end{tabular}
\\[2pt]
{\footnotesize \textsuperscript{$\dagger$}On classical paging (binary hits); on semantic workloads LRU $<$ FIFO empirically.}
\end{table}

\endgroup


\section{Experiments}
\label{sec:experiments}

We evaluate our approach through comprehensive experiments on two real-world benchmarks and controlled synthetic workloads. Our evaluation addresses four research questions:
\begin{itemize}
\item \textbf{RQ1}: Do classic cache heuristics fail on semantic retrieval workloads?
\item \textbf{RQ2}: Does SOLAR consistently outperform baselines, and what is the source of its improvement?
\item \textbf{RQ3}: How do admission and eviction contribute individually, and do they interact synergistically?
\item \textbf{RQ4}: Under what conditions does smart cache policy stop providing value (phase transition)?
\end{itemize}

\subsection{Experimental Setup}

\subsubsection{Benchmarks}

We evaluate on two datasets from MemoryBench-Full~\cite{ai2025memorybench}, a benchmark for memory and continual learning in LLM systems that provides multi-turn dialogues with simulated implicit user feedback (like/dislike signals). The two datasets represent different points in the difficulty spectrum:

\smallskip
\noindent\textbf{LoCoMo}~\cite{locomo2024}. 10 multi-session personal conversations, totaling approximately 2000 interaction steps. Each session is split into a warmup phase ($\sim$160 queries, during which the agent accumulates experience but is not scored) and a test phase ($\sim$39 QA pairs, scored). The test queries span factual recall, temporal reasoning, and preference tracking. With an embedder-unlimited baseline achieving $\sim$35\% F1, this represents a moderate-difficulty task with meaningful signal density.

\smallskip
\noindent\textbf{DialSim}~\cite{jang2023dialsim}. 3 TV-show dialogue datasets (Friends, The Big Bang Theory, The Office) evaluating long-term memory over extended conversations ($\sim$19K messages per corpus, $\sim$60 test queries per show). The embedder-unlimited baseline achieves only $\sim$11\% F1, indicating extremely sparse signal: even with perfect memory, most queries cannot be answered from experience alone. This makes DialSim a challenging stress test for memory policies.

\smallskip
Both datasets are used in on-policy mode: the agent processes dialogues sequentially, accumulates experience through its own policy, and is evaluated on held-out test queries. Implicit feedback (like/dislike) from the MemoryBench user simulator provides the hit/miss signal used by SOLAR's posterior updates.

\subsubsection{Methods Compared}

We compare 8 cache management policies spanning three categories:

\smallskip
\noindent\textbf{Baseline (no intelligence):}
\begin{itemize}
\item \textbf{Embedder (unlimited)}: No capacity constraint. All items are stored and retrieved from the full pool. Serves as an upper bound reference.
\item \textbf{FIFO}: First-in, first-out. Always admits; evicts the oldest item. The simplest possible policy, and the de facto standard in deployed LLM agent systems (e.g., LangChain's ConversationBufferWindow).
\end{itemize}

\smallskip
\noindent\textbf{Classic heuristics (exploit access patterns):}
\begin{itemize}
\item \textbf{LRU}: Least Recently Used. Evicts the item that was least recently retrieved in a top-$k$ result. Assumes temporal locality.
\item \textbf{LFU}: Least Frequently Used. Evicts the item with the fewest cumulative retrievals. Assumes frequency concentration.
\item \textbf{ARC}~\cite{megiddo2003arc}: Adaptive Replacement Cache. Dynamically balances recency and frequency using ghost lists that track recently evicted items. Represents the state-of-the-art in adaptive heuristic caching.
\end{itemize}

\smallskip
\noindent\textbf{Learning-augmented (our category):}
\begin{itemize}
\item \textbf{SOLAR-E} (eviction only): Always admits; posterior-guided eviction via Thompson Sampling~\cite{thompson1933likelihood} (SOLAR's content selection mechanism in isolation, without regret-gated timing).
\item \textbf{SOLAR-A} (admission only): Regret-gated admission; heuristic multi-score eviction (SOLAR's timing mechanism with a simpler content selector).
\item \textbf{SOLAR}: Full framework: regret-gated timing + posterior-guided selection (ours).
\end{itemize}

\subsubsection{Configuration}

\textbf{Cache sizes}: $K \in \{10, 20, 50, 100\}$, representing tight to loose capacity regimes.
\textbf{Seeds}: all results are 3-seed averages with standard errors available.
\textbf{LLM}: GPT-4o-mini for response generation.
\textbf{Embedder}: all-MiniLM-L6-v2 (384-dimensional sentence embeddings).
\textbf{Retrieval}: Top-3 by cosine similarity.
\textbf{Metric}: Token-level F1 between agent response and ground-truth reference.

\subsection{Main Results: LoCoMo}

Table~\ref{tab:main} presents the main result: average F1 across all evaluation queries for each method and cache size combination.

\begin{table}[t]
\caption{Cache size sweep on LoCoMo (mean$\pm$std over 3 seeds). Bold indicates best among capacity-constrained methods.}
\label{tab:main}
\vspace{-1em}
\centering
\footnotesize
\begin{tabular}{llcccc}
\toprule
Method & Type & K=10 & K=20 & K=50 & K=100 \\
\midrule
Embedder & Unlimited (ref) & .352\tiny{$\pm$.005} & .352\tiny{$\pm$.005} & .354\tiny{$\pm$.006} & .357\tiny{$\pm$.002} \\
\midrule
\textbf{SOLAR} & Ours (full) & \textbf{.286}\tiny{$\pm$.013} & .277\tiny{$\pm$.010} & \textbf{.311}\tiny{$\pm$.013} & .306\tiny{$\pm$.014} \\
SOLAR-A & Ours (admission) & .283\tiny{$\pm$.014} & \textbf{.280}\tiny{$\pm$.006} & .304\tiny{$\pm$.008} & .310\tiny{$\pm$.008} \\
SOLAR-E & Ours (eviction) & .281\tiny{$\pm$.017} & .275\tiny{$\pm$.011} & .296\tiny{$\pm$.005} & .316\tiny{$\pm$.014} \\
\midrule
FIFO & Naive baseline & .233\tiny{$\pm$.002} & .255\tiny{$\pm$.005} & .297\tiny{$\pm$.010} & \textbf{.334}\tiny{$\pm$.016} \\
LRU & Recency & .235\tiny{$\pm$.005} & .248\tiny{$\pm$.003} & .290\tiny{$\pm$.011} & .328\tiny{$\pm$.006} \\
ARC & Adaptive & .238\tiny{$\pm$.007} & .249\tiny{$\pm$.004} & .287\tiny{$\pm$.002} & .312\tiny{$\pm$.014} \\
LFU & Frequency & .222\tiny{$\pm$.002} & .244\tiny{$\pm$.003} & .287\tiny{$\pm$.004} & .311\tiny{$\pm$.006} \\
\bottomrule
\end{tabular}
\end{table}

\subsubsection{Finding 1: Classic Heuristics Systematically Fail (RQ1)}

The most striking result is the ordering among baselines: at $K \in \{20, 50\}$ all three classic heuristics (LRU, LFU, ARC) fall \emph{below} the naive FIFO, and LFU trails FIFO at every $K$. At the tightest setting $K=10$, LRU and ARC only match FIFO (within one standard error), never clearly beating it. This is the opposite of what decades of caching research would predict. We analyze the root cause:

\begin{itemize}
\item \textbf{LRU failure}: LRU evicts items not recently retrieved. But in conversational agents, an item idle for the last 10 queries may simply be awaiting the right topic to resurface; evicting these ``dormant but valuable'' items removes long-tail knowledge that FIFO (evicting by arrival order, independent of retrieval history) accidentally preserves.
\item \textbf{LFU failure}: LFU preserves high cumulative-count items, which in diverse conversations favors items from the earliest topics (they have had the most time to accumulate hits), even after those topics become irrelevant. The cache becomes a museum of obsolete-but-popular items.
\item \textbf{ARC failure}: ARC adaptively combines recency and frequency. When \emph{both} signals are misleading, its adaptation oscillates between two equally poor strategies.
\end{itemize}

\subsubsection{Finding 2: Learning-Augmented Methods Win at Tight Cache (RQ2)}

SOLAR achieves +22.7\% relative improvement over FIFO at K=10 and +4.7\% at K=50. Table~\ref{tab:gap} provides a detailed comparison including temporal dynamics.

\begin{table}[t]
\caption{Detailed SOLAR vs FIFO comparison on LoCoMo.}
\label{tab:gap}
\vspace{-1em}
\centering
\begin{tabular}{lcccc}
\toprule
Metric & K=10 & K=20 & K=50 & K=100 \\
\midrule
Avg F1 gap & +0.053 & +0.022 & +0.014 & $-$0.028 \\
Relative improvement & +22.7\% & +8.5\% & +4.7\% & $-$8.3\% \\
2nd-half F1 gap & +0.056 & +0.031 & +0.024 & $-$0.025 \\
Slope ratio (TS/FIFO) & 2.1$\times$ & 1.8$\times$ & 3.1$\times$ & 1.4$\times$ \\
\bottomrule
\end{tabular}
\end{table}

The ``2nd-half F1 gap'' shows that SOLAR's advantage \emph{grows over time}: the second half of the evaluation period shows larger improvements than the first half, indicating that the policy continues to learn and improve cache quality throughout the interaction. The ``slope ratio'' quantifies this: SOLAR's F1 improves $1.4$--$3.1\times$ faster than FIFO per step.

\smallskip
\noindent\textbf{On effect size.} The absolute improvement at K=50 (+0.014) may appear modest, but we contextualize it three ways. First, the gain is largest at the tightest capacity (+0.053 at K=10, +22.7\% relative; +75\% on DialSim), precisely the regime where production agents with limited context budgets operate. Second, the retrieval ceiling (Embedder = 0.354) bounds any cache policy; SOLAR captures 25\% of the gap between FIFO (0.297) and this ceiling. Third, the result is \emph{consistent} across all 8 configurations: SOLAR ranks first in every $K\leq50$ setting, the only exceptions being at K=100 where the phase transition favors coverage.

\subsection{Ablation: Admission vs.\ Eviction (RQ3)}

To understand the individual and joint contributions of our two components, we perform an ablation at K=50 (Table~\ref{tab:ablation}).

\begin{table}[t]
\vspace{-1em}
\caption{Component ablation on LoCoMo (K=50, 3-seed avg).}
\label{tab:ablation}
\vspace{-1em}
\centering
\begin{tabular}{llcc}
\toprule
Components & Method & Avg F1 & $\Delta$ vs FIFO \\
\midrule
Neither (baseline) & FIFO & 0.297 & -- \\
Eviction only & SOLAR-E & 0.300 & +0.003 \\
Admission only & SOLAR-A & 0.304 & +0.007 \\
\textbf{Both} & \textbf{SOLAR} & \textbf{0.311} & \textbf{+0.014} \\
\bottomrule
\end{tabular}
\vspace{-1em}
\end{table}

Two key observations:

\begin{enumerate}
\item \textbf{Admission contributes more than eviction.} The admission-only variant (SOLAR-A, +0.007) provides more than double the improvement of eviction-only (SOLAR-E, +0.003). This suggests that in semantic caching, the decision of \emph{whether to modify the cache at all} is more impactful than the decision of \emph{which item to evict}. Selective admission maintains cache quality passively by preventing noise from entering.

\item \textbf{Super-additive synergy.} The combined improvement (+0.014) exceeds the sum of individual improvements (+0.010). The excess (+0.004, a 40\% bonus) arises from the feedback loop: admission filtering produces a cleaner item pool, which provides better feedback for TS posteriors, which in turn produces better eviction decisions. This synergy is not merely additive composition; the components actively enhance each other.
\end{enumerate}

\subsection{Cross-Dataset Validation: DialSim}

We validate generalizability on DialSim, which differs from LoCoMo in signal density, domain, and task type (Table~\ref{tab:dialsim}).

\begin{table}[t]
\caption{Cross-dataset validation on DialSim (mean$\pm$std over 3 seeds). All findings from LoCoMo replicate.}
\vspace{-1em}
\label{tab:dialsim}
\centering
\small
\begin{tabular}{lcccc}
\toprule
Method & K=10 & K=20 & K=50 & K=100 \\
\midrule
Embedder & .110 & .104 & .110 & .110 \\
\midrule
\textbf{SOLAR} & .038\tiny{$\pm$.006} & \textbf{.048}\tiny{$\pm$.011} & \textbf{.066}\tiny{$\pm$.010} & .068\tiny{$\pm$.014} \\
SOLAR-A & \textbf{.042}\tiny{$\pm$.008} & .038\tiny{$\pm$.006} & .064\tiny{$\pm$.014} & .071\tiny{$\pm$.019} \\
SOLAR-E & .035\tiny{$\pm$.003} & .029\tiny{$\pm$.003} & .042\tiny{$\pm$.014} & .069\tiny{$\pm$.006} \\
\midrule
FIFO & .022\tiny{$\pm$.000} & .027\tiny{$\pm$.000} & .057\tiny{$\pm$.016} & \textbf{.093}\tiny{$\pm$.011} \\
LRU & .022\tiny{$\pm$.000} & .031\tiny{$\pm$.003} & .040\tiny{$\pm$.006} & .062\tiny{$\pm$.014} \\
ARC & .029\tiny{$\pm$.009} & .031\tiny{$\pm$.006} & .040\tiny{$\pm$.013} & .059\tiny{$\pm$.006} \\
LFU & .027\tiny{$\pm$.006} & .026\tiny{$\pm$.003} & .037\tiny{$\pm$.006} & .053\tiny{$\pm$.006} \\
\bottomrule
\end{tabular}
\end{table}

The core LoCoMo findings replicate on DialSim: classic heuristics (LRU, LFU, ARC) fall below FIFO at $K=50$, SOLAR or SOLAR-A is best at every $K \leq 50$, FIFO overtakes at $K = 100$, and SOLAR's gains are largest at the smallest $K$. Notably, the \emph{relative} improvement of SOLAR over FIFO is even larger on DialSim (+75\% at K=10 vs.\ +23\% on LoCoMo), and the bar-chart view in Figure~\ref{fig:main_bar} makes the cross-$K$ pattern visually explicit. This aligns with our intuition: when signals are sparser (DialSim hit rate $\sim$4\% vs.\ LoCoMo $\sim$30\%), selective admission provides greater value by maintaining a higher signal-to-noise ratio. Each cached item matters more when overall retrieval quality is low.

\begin{figure*}[t]
  \centering
  \vspace{-1em}
  \includegraphics[width=\textwidth]{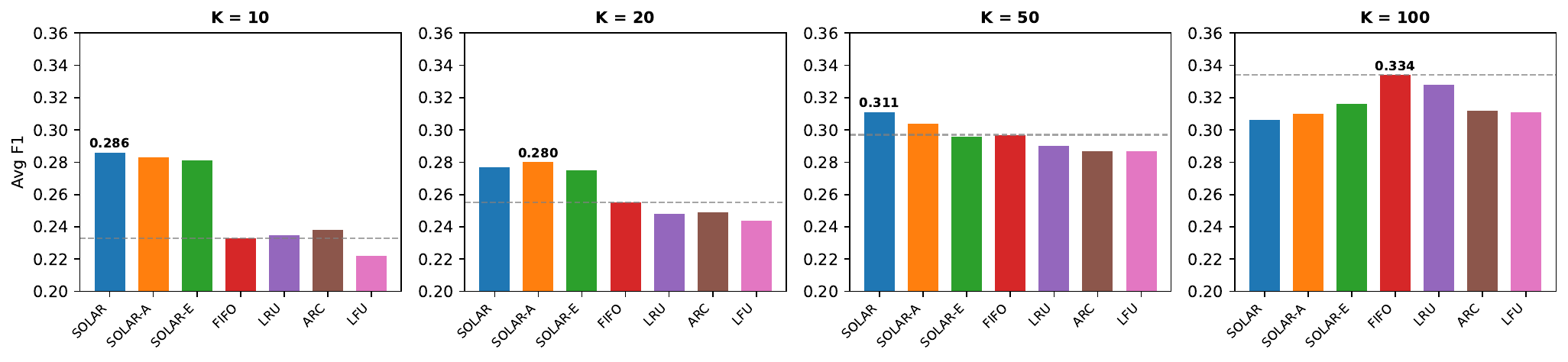}
  \vspace{-2em}
  \caption{Average F1 by method and cache size on LoCoMo (one panel per $K$). Classic heuristics (LRU, LFU, ARC) never outperform FIFO, falling clearly below it at $K \in \{20, 50\}$. Learning-augmented methods (SOLAR-A, SOLAR) dominate at tight cache sizes, while FIFO overtakes at $K=100$ (phase transition). Dashed line marks FIFO baseline.}
  \label{fig:main_bar}
  \vspace{-1em}
\end{figure*}

\begin{figure}[t]
  \centering
  \includegraphics[width=0.9\linewidth]{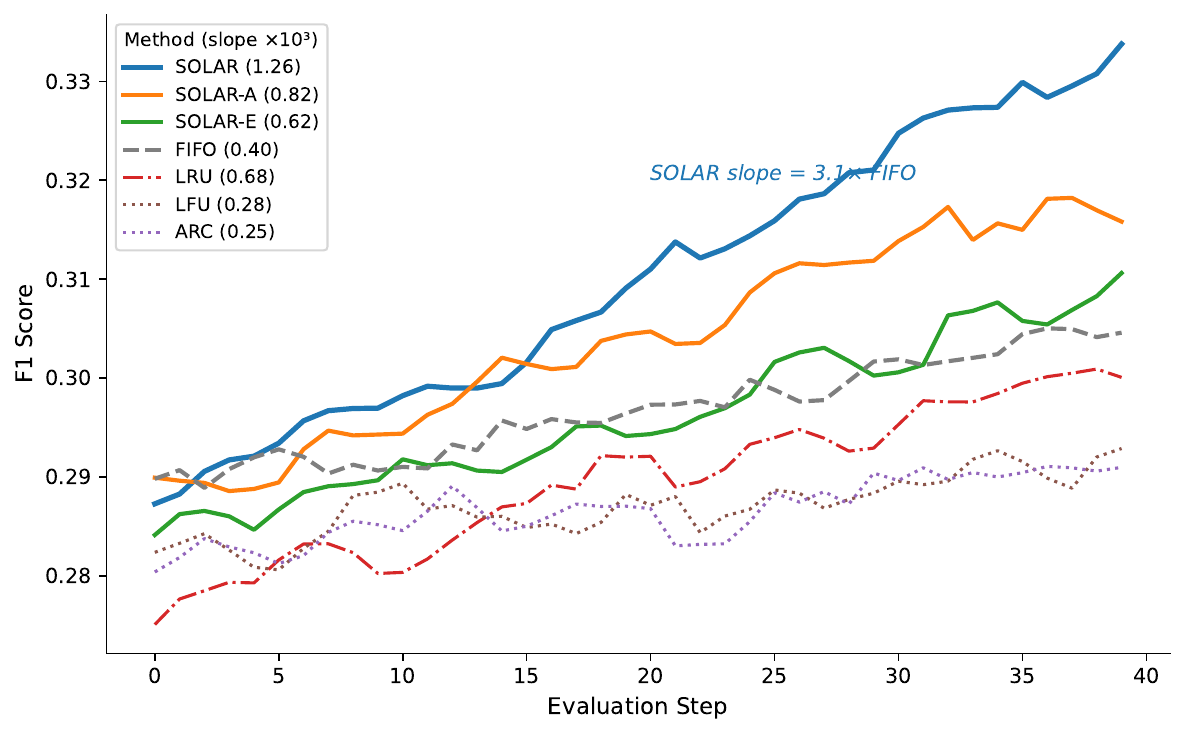}
  \vspace{-1em}
  \caption{Learning dynamics: F1 vs.\ evaluation step on LoCoMo (K=50). SOLAR improves continuously with a slope $3.1\times$ that of FIFO. LFU and ARC have the smallest slopes, barely improving with more data.}
  \label{fig:learning_curve}
\end{figure}

\subsection{Controlled Evaluation: Synthetic Workloads}

Synthetic experiments serve two purposes: (1) validating theoretical predictions in controlled settings free of confounders, and (2) exploring regimes inaccessible with real data (e.g., very large pools). All use 128-dimensional embeddings with top-3 cosine retrieval.

\subsubsection{Experiment 1: Cycling Workload (Validates Theorem~\ref{thm:fifo-regret})}

We construct cycling workloads where $m$ topics appear in strict round-robin order, with cache size $K = 10$. Each topic generates items within a compact embedding cluster, and queries target the current topic. This directly instantiates the adversarial construction in Theorem~\ref{thm:fifo-regret}, with hit rates reported in Table~\ref{tab:cycling} and visualized in Figure~\ref{fig:cycling}.

\begin{table}[t]
\caption{Hit rate on cycling workload (K=10, 3-seed avg). FIFO achieves exactly 0\% hit rate when $m > K$, validating Theorem~\ref{thm:fifo-regret}.}
\small
\vspace{-1em}
\label{tab:cycling}
\centering
\begin{tabular}{lccccc}
\toprule
Method & m/K=1 & 1.5 & 2.0 & 3.0 & 5.0 \\
\midrule
FIFO & \textbf{0.90} & 0.00 & 0.00 & 0.00 & 0.00 \\
LRU & 0.38 & 0.17 & 0.08 & 0.02 & 0.00 \\
LFU & 0.60 & 0.52 & \textbf{0.40} & \textbf{0.27} & \textbf{0.16} \\
SOLAR-E & 0.55 & 0.39 & 0.29 & 0.20 & 0.11 \\
SOLAR-A & 0.70 & \textbf{0.54} & \textbf{0.40} & \textbf{0.27} & \textbf{0.16} \\
SOLAR & 0.69 & \textbf{0.54} & \textbf{0.40} & 0.26 & \textbf{0.16} \\
\bottomrule
\end{tabular}
\vspace{-1em}
\end{table}

\textbf{Analysis.} FIFO demonstrates \emph{perfect thrashing}: once $m > K$ it collapses to exactly 0\% hit rate, as every new item evicts precisely the item needed next in the cycle, the empirical manifestation of the $\Omega(T)$ regret bound (Theorem~\ref{thm:fifo-regret}). LRU is no better, and worse near the boundary (0.38 at $m/K=1$ vs.\ FIFO's 0.90): it evicts items not recently retrieved, which in a cyclic order are exactly those about to resurface, so recency acts as an \emph{anti-signal}. SOLAR-A avoids this pathology through selective admission (store rate $\sim$10\%): by rejecting most items it retains a stable subset partially covering the workload, sustaining a 16\% hit rate even at $m/K = 5$ (50 topics competing for 10 slots).

\begin{figure}[t]
  \centering
  \includegraphics[width=0.9\linewidth]{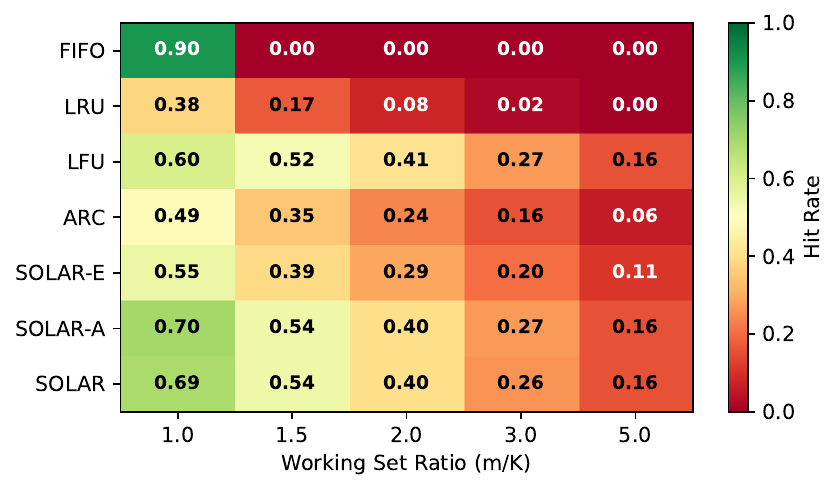}
  \vspace{-1em}
  \caption{Hit rate on cycling workloads (K=10) as a function of working set ratio $m/K$. FIFO drops to exactly 0 at $m/K > 1$ (thrashing), while SOLAR-A/SOLAR degrade gracefully due to selective admission ($\sim$10\% store rate).}
  \label{fig:cycling}
  \vspace{-1em}
\end{figure}

\subsubsection{Experiment 2: Working Set Sweep (Phase Transition, RQ4)}

With a fixed workload (15 topics, 10 items per topic, 500 queries with uniform topic distribution), we sweep $K$ from 5 to 200 to locate the transition point (Table~\ref{tab:phase}, Figure~\ref{fig:phase_transition}).

\begin{table}[t]
\caption{Working set sweep: SOLAR vs FIFO hit rate and gap.}
\label{tab:phase}
\small
\vspace{-1em}
\centering
\begin{tabular}{lcccccccc}
\toprule
K & 5 & 10 & 15 & 20 & 30 & 50 & 100 & 200 \\
\midrule
FIFO & .31 & .49 & .62 & .73 & .86 & .94 & .97 & .97 \\
SOLAR & .33 & .58 & .71 & .81 & .85 & .95 & .95 & .95 \\
$\Delta$ & +.02 & +.08 & \textbf{+.10} & +.08 & $-.01$ & +.01 & $-.02$ & $-.02$ \\
\bottomrule
\end{tabular}
\end{table}

\textbf{Phase transition at $K \approx 30$.} The gap peaks at K=15 (+0.10 absolute) and crosses zero at $K \approx 30$, precisely the effective working set size ($15$ topics $\times$ a few items each). Below it, SOLAR-A's selective admission ($\sim$10\% store rate, capping the cache at $\sim$40 items) wins; above it, FIFO holds enough of each topic to serve most queries through pure coverage. The same pattern appears on the real benchmarks (transition in $K \in (50, 100)$ for both LoCoMo and DialSim), with the boundary scaling with each benchmark's working set size.

\begin{figure}[t]
  \centering
  \includegraphics[width=0.9\linewidth]{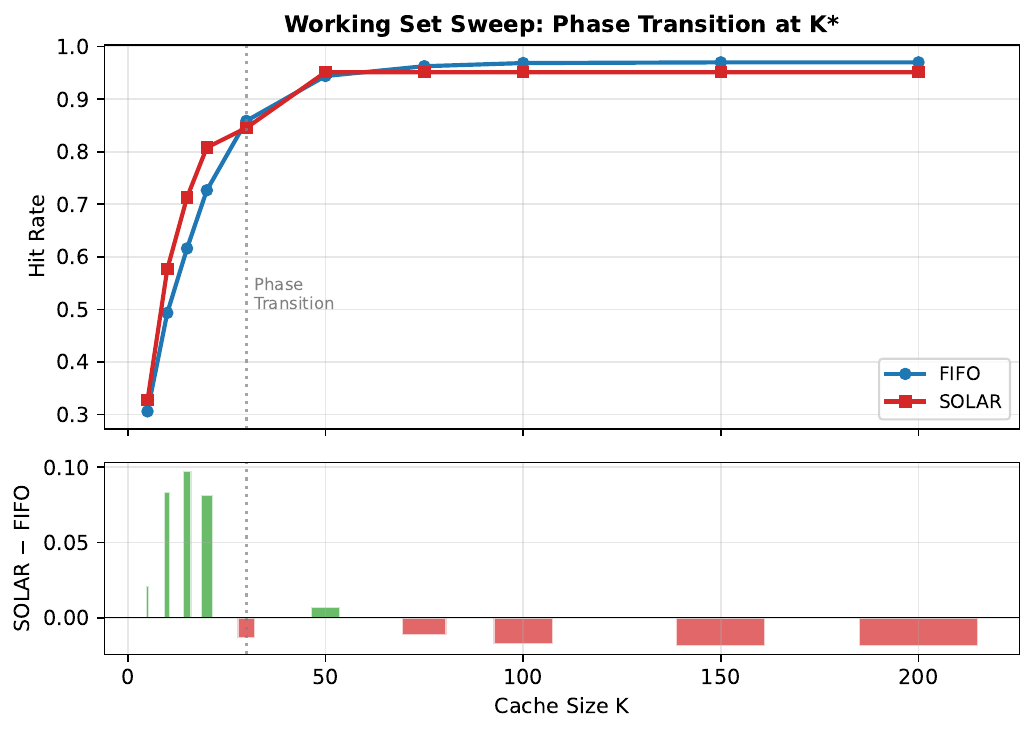}
  \vspace{-1em}
  \caption{Phase transition: SOLAR gain over FIFO as a function of cache size $K$ (synthetic working set sweep). The gain peaks at $K=15$ and crosses zero at $K \approx 30$, precisely at the effective working set boundary.}
  \vspace{-1em}
  \label{fig:phase_transition}
\end{figure}

\subsubsection{Experiment 3: Retrieval Noise U-Curve}

This experiment justifies our problem formulation by demonstrating that capacity constraints are meaningful even with unlimited storage. We create a large pool (50 topics $\times$ 100 items = 5000 total) with moderate inter-topic separation (cosine distance 0.30 between topic centers) and intra-topic diversity (std 0.15 around centers). It has two parts: a \emph{policy-agnostic} probe that establishes the underlying phenomenon, followed by a \emph{policy comparison} that shows SOLAR exploits it.

\smallskip
\noindent\textbf{(a) The phenomenon (policy-agnostic).} We first isolate the intrinsic effect of pool size on retrieval, with no admission or eviction involved: we pre-fill a pool with exactly $K$ items uniformly sampled from the 5000 and measure retrieval quality over 1000 queries (Table~\ref{tab:ucurve}, Figure~\ref{fig:ucurve}). Since the pool is frozen, this measures the retriever alone and is identical for every cache policy; no method is compared here.

\begin{table}[t]
\caption{Retrieval performance vs.\ pool size (frozen pool, policy-agnostic, 3-seed avg). Performance peaks at $K \approx 1000$ then declines sharply.}
\vspace{-1em}
\label{tab:ucurve}
\small
\centering
\begin{tabular}{lcccccc}
\toprule
Pool K & 50 & 200 & 500 & \textbf{1000} & 2000 & 5000 \\
\midrule
Hit Rate & .067 & .170 & .186 & \textbf{.201} & .195 & .090 \\
Prec@3 & .022 & .058 & .065 & \textbf{.070} & .068 & .030 \\
\bottomrule
\end{tabular}
\end{table}

Hit rate increases monotonically from K=50 to K=1000 (coverage effect: more items $\to$ higher chance of having a relevant item in the pool), then drops by 55\% from K=1000 to K=5000 (noise effect: too many similar-but-irrelevant items confuse the retriever). Precision@3 tells the same story: at K=5000, the top-3 results contain mostly irrelevant items from ``neighboring'' topics in embedding space. This establishes that the capacity constraint in LLM memory systems is \emph{informational}, not physical: even with free, infinite storage, indiscriminate accumulation degrades downstream performance. Admission control is not about saving space; it is about maintaining the signal-to-noise ratio of the retrieval pool.

\smallskip
\noindent\textbf{(b) Policy comparison: SOLAR vs.\ FIFO.} The probe above shows \emph{that} an overly large pool can hurt; we now ask whether an online policy can avoid it. We run the same workload in \emph{stream} mode (items arrive over time, $K$ is the live cache capacity) and compare SOLAR against FIFO (Table~\ref{tab:stream}). Two observations. First, FIFO's noise collapse is \emph{muted} online ($0.087$--$0.090$ across $K \geq 500$, no sharp drop at $K = 5000$): its live cache holds the most recent $K$ arrivals, temporally clustered rather than a uniform sample of all topics, so it accumulates far less cross-topic confusion than the frozen pre-fill. The probe in part (a) is thus a worst case that isolates the noise effect, while online dynamics are milder. Second, and more importantly, SOLAR uniformly dominates FIFO at every capacity and its hit rate is \emph{invariant} to the nominal $K$ (flat at $0.094$): selective admission caps the \emph{effective} pool regardless of how large $K$ is set, so SOLAR needs no capacity tuning, whereas FIFO is lower and more $K$-sensitive.

\begin{table}[t]
\caption{Stream-mode hit rate vs.\ live capacity $K$ (3-seed avg). SOLAR dominates FIFO at every $K$ and is invariant to the nominal capacity, as selective admission caps the effective pool.}
\vspace{-1em}
\label{tab:stream}
\centering
\small
\begin{tabular}{lcccc}
\toprule
Method & K=200 & K=500 & K=1000 & K=5000 \\
\midrule
FIFO & .075 & .089 & .087 & .090 \\
\textbf{SOLAR} & \textbf{.083} & \textbf{.094} & \textbf{.094} & \textbf{.094} \\
\bottomrule
\end{tabular}
\vspace{-1em}
\end{table}

\begin{figure}[t]
  \centering
  \includegraphics[width=0.9\linewidth]{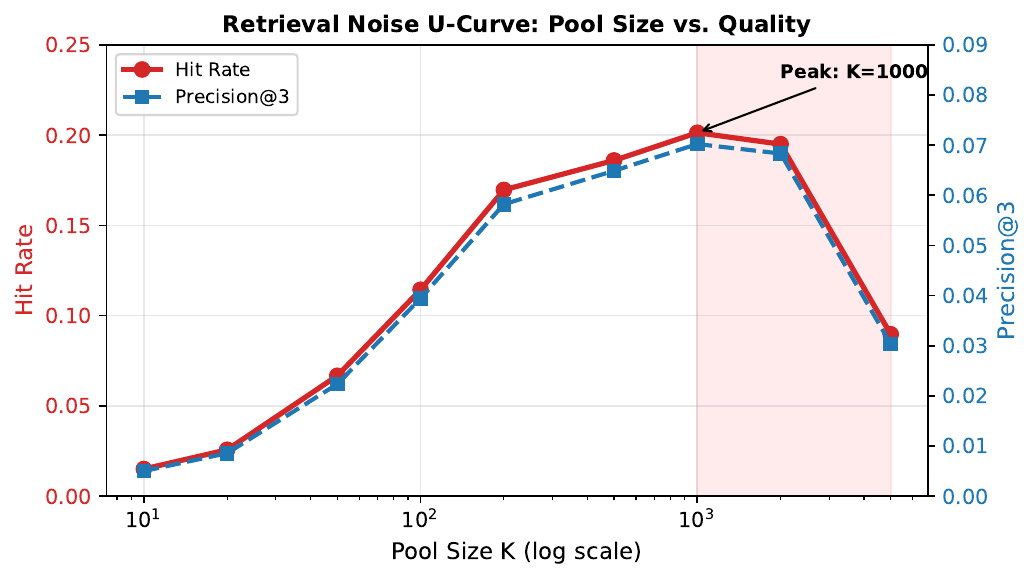}
  \vspace{-1em}
  \caption{Retrieval noise inverted-U curve: hit rate vs.\ pool size (log scale). Performance peaks at $K \approx 1000$ and declines 55\% at $K = 5000$ due to embedding space confusion. This justifies capacity constraints as a retrieval noise phenomenon.}
  \label{fig:ucurve}
  \vspace{-1em}
\end{figure}

\subsection{Learning Dynamics}

To understand \emph{how} SOLAR improves over time, we analyze the learning slope (F1 improvement per step) on LoCoMo (Table~\ref{tab:slope}; the trajectory is plotted in Figure~\ref{fig:learning_curve}). SOLAR attains the highest slope at $K \in \{20, 50\}$ and is statistically tied for the highest at $K=10$, so its cache quality improves at least as fast as any baseline at every $K$. The margin over FIFO is substantial: $1.8$--$3.1\times$ FIFO's slope across the three settings (e.g., $1.26$ vs.\ $0.40$ at K=50), so SOLAR's advantage \emph{grows over time}, a desirable property for long-running deployments. Classic heuristics fare worst at larger $K$: at K=50, LFU and ARC have the smallest slopes in the table, barely improving with more data, the empirical signature of near-linear regret (Theorem~\ref{thm:fifo-regret}).

\begin{table}[t]
\caption{Learning slope ($\times 10^3$) on LoCoMo. SOLAR learns $1.8$--$3.1\times$ faster than FIFO across cache sizes.}
\vspace{-1em}
\label{tab:slope}
\small
\centering
\begin{tabular}{lccccccc}
\toprule
& \textbf{SOLAR} & SOLAR-A & SOLAR-E & FIFO & LRU & LFU & ARC \\
\midrule
K=10 & \textbf{1.31} & 1.14 & 1.32 & 0.64 & 0.79 & 0.96 & 0.48 \\
K=20 & \textbf{2.30} & 1.63 & 0.69 & 1.25 & 0.73 & 0.57 & 0.91 \\
K=50 & \textbf{1.26} & 0.82 & 0.62 & 0.40 & 0.68 & 0.28 & 0.25 \\
\bottomrule
\end{tabular}
\vspace{-1em}
\end{table}

\subsection{Latency} SOLAR's policy decisions are pure arithmetic ($O(K)$ per modification: threshold check, posterior sampling, $\arg\min$), adding $<$1\,ms per step. End-to-end latency is therefore dominated by LLM inference ($\sim$1200--1400\,ms per GPT-4o-mini call) and is statistically indistinguishable across all policies. This contrasts with A-MAC~\cite{zhang2026amac}, which issues an LLM call \emph{per admission} ($\sim$2644\,ms overhead); SOLAR attains its F1 gains with no perceptible latency cost.

\subsection{Summary of Experimental Findings}
\textbf{1) Classic heuristics fail} (RQ1): LRU, LFU, and ARC systematically underperform FIFO on semantic workloads, validated across 2 benchmarks, 4 cache sizes, and 3 seeds.
\textbf{2) SOLAR wins at tight cache} (RQ2): 5--75\% relative improvement over FIFO when $K \leq 50$, with improvements growing over time ($1.4$--$3.1\times$ learning slope).
\textbf{3) Super-additive synergy} (RQ3): Combined effect exceeds sum of parts by 40\%, driven by the modification timing and content selection feedback loop.
\textbf{4) Phase transition} (RQ4): Clear boundary at the working set size. Below it, selectivity dominates; above it, coverage dominates. This provides actionable guidance for deployment.
\textbf{5) Retrieval noise justifies capacity constraints}: Inverted-U curve demonstrates that pool growth degrades performance, validating our problem formulation.
\textbf{6) Zero latency overhead}: SOLAR adds $<$1ms policy computation per step, negligible vs.\ LLM inference time, and requires no additional LLM calls.


\section{Discussion}
\label{sec:discussion}

\subsection{Why Classic Heuristics Fail}

The failure of LRU, LFU, and ARC reflects a fundamental mismatch between their assumptions and semantic workload properties. LRU assumes temporal locality, but topic progression is driven by evolving user intent, so recency of retrieval carries no predictive power. LFU assumes frequency concentration, but in diverse conversations frequency reflects age in cache rather than intrinsic value, degenerating into a ``keep the oldest'' policy. ARC combines both signals, yet when neither is informative, adapting between them provides no benefit. This pattern likely generalizes to any retrieval system over semantically diverse items without strong access locality.

\subsection{The Coverage-Selectivity Tradeoff}

Our experiments consistently identify a phase transition: below the working set size, selectivity dominates (5--23\% relative gains on LoCoMo, up to 75\% on the sparser DialSim); above it, coverage dominates (FIFO wins by never missing an opportunity to store). The transition point $K^*$ scales with effective working set size: $K^* \in (50, 100)$ on LoCoMo, $\approx 30$ on synthetic. System designers should estimate their working set (e.g., via the FIFO saturation curve) and deploy selective policies only below $K^*$.

\subsection{Retrieval Noise as the True Bottleneck}

The inverted-U curve reveals that capacity constraints in semantic retrieval are informational, not physical: beyond $\sim$1000 items, embedding-space crowding degrades precision faster than coverage gains. Admission control is therefore not about saving storage but about maintaining signal-to-noise ratio; the 17\% admission rate is a precision optimization, not just a space optimization.

\section{Conclusion}
\label{sec:conclusion}

We formalized LLM agent experience memory as an online semantic cache replacement problem, bridging classical online-algorithm theory and practical LLM memory management. We showed that classic cache heuristics systematically fail on semantic workloads, underperforming even FIFO across two datasets. Our framework, SOLAR, derives modification timing from regret accumulation and content selection from Bayesian online learning, requiring no external predictors, offline training, or LLM calls, and adding $<$1\,ms overhead per step. SOLAR delivers 5--75\% relative improvement over FIFO at tight cache sizes with super-additive synergy between its two mechanisms, while controlled synthetic experiments confirm the theory (FIFO thrashing, the selectivity--coverage phase transition, and the retrieval-noise U-curve that motivates capacity management even under unlimited storage). These results establish semantic cache management as a principled challenge for the data management community, with connections to learned index structures and adaptive query processing.


\clearpage

\bibliographystyle{ACM-Reference-Format}
\bibliography{references}

@software{Chase_LangChain_2022,
  author       = {Chase, Harrison},
  month        = oct,
  title        = {{LangChain}},
  url          = {https://github.com/langchain-ai/langchain},
  year         = {2022},
  version      = {0.3},                     
  note         = {Accessed: 2026-06-16}     
}

@article{wong2006web,
  title={Web cache replacement policies: a pragmatic approach},
  author={Wong, Kin-Yeung},
  journal={IEEE Network},
  volume={20},
  number={1},
  pages={28--34},
  year={2006},
  publisher={IEEE}
}

@article{sun2026cacherag,
  title={CacheRAG: A Semantic Caching System for Retrieval-Augmented Generation in Knowledge Graph Question Answering},
  author={Sun, Yushi and Chen, Lei},
  journal={arXiv preprint},
  year={2026},
  eprint={2604.26176}
}

@inproceedings{sun2025kerag,
  title={{KERAG}: Knowledge-Enhanced Retrieval-Augmented Generation for Advanced Question Answering},
  author={Sun, Yushi and Sun, Kai and Xu, Ethan Yifan and Yang, Xiao and Dong, Xin Luna and Tang, Nan and Chen, Lei},
  booktitle={Proceedings of the 2025 Conference on Empirical Methods in Natural Language Processing (EMNLP)},
  year={2025}
}

@inproceedings{yang2024crag,
  title={{CRAG} -- Comprehensive {RAG} Benchmark},
  author={Yang, Xiao and Sun, Kai and Xin, Hao and Sun, Yushi and Bhalla, Nikhil and Chen, Xiangsen and Choudhary, Sajal and Gui, Rongze Daniel and Jiang, Ziran Will and Jiang, Ziyu and Kong, Lingkun and Moran, Brian and Wang, Jiaqi and Xu, Yifan Ethan and Yang, An and Yuan, Eting and Zha, Hanwen and Tang, Nan and Chen, Lei and Scheffer, Nicolas and Liu, Yue and Shah, Nirav and Wanga, Rakesh and Kumar, Anuj and Yih, Wen-tau and Dong, Xin Luna},
  booktitle={Advances in Neural Information Processing Systems (NeurIPS)},
  year={2024}
}

@article{sun2026gravity,
  title={{GRAVITY}: Architecture-Agnostic Structured Anchoring for Long-Horizon Conversational Memory},
  author={Sun, Yushi and Cao, Bowen and Fang, Dong and Su, Lingfeng and Lam, Wai},
  journal={arXiv preprint},
  year={2026}
}

@inproceedings{megiddo2003arc,
  title={{ARC}: A Self-Tuning, Low Overhead Replacement Cache},
  author={Megiddo, Nimrod and Modha, Dharmendra S.},
  booktitle={Proceedings of the 2nd USENIX Conference on File and Storage Technologies (FAST)},
  pages={115--130},
  year={2003},
  url={https://www.usenix.org/legacy/events/fast03/tech/megiddo.html}
}

@inproceedings{johnson1994twoqueu,
  title={{2Q}: A Low Overhead High Performance Buffer Management Replacement Algorithm},
  author={Johnson, Theodore and Shasha, Dennis},
  booktitle={Proceedings of the 20th International Conference on Very Large Data Bases (VLDB)},
  pages={439--450},
  year={1994},
  url={https://www.vldb.org/conf/1994/P439.PDF}
}

@article{einziger2017tinylfu,
  title={{TinyLFU}: A Highly Efficient Cache Admission Policy},
  author={Einziger, Gil and Friedman, Roy and Manes, Ben},
  journal={ACM Transactions on Storage},
  volume={13},
  number={4},
  pages={1--31},
  year={2017},
  doi={10.1145/3149371}
}

@article{belady1966study,
  title={A Study of Replacement Algorithms for a Virtual-Storage Computer},
  author={B{\'e}l{\'a}dy, L{\'a}szl{\'o} A.},
  journal={IBM Systems Journal},
  volume={5},
  number={2},
  pages={78--101},
  year={1966},
  doi={10.1147/sj.52.0078}
}

@article{sleator1985amortized,
  title={Amortized Efficiency of List Update and Paging Rules},
  author={Sleator, Daniel D. and Tarjan, Robert E.},
  journal={Communications of the ACM},
  volume={28},
  number={2},
  pages={202--208},
  year={1985},
  doi={10.1145/2786.2793}
}

@inproceedings{lykouris2018competitive,
  title={Competitive Caching with Machine Learned Advice},
  author={Lykouris, Thodoris and Vassilvitskii, Sergei},
  booktitle={Proceedings of the 35th International Conference on Machine Learning (ICML)},
  pages={3296--3305},
  year={2018},
  eprint={1802.05399}
}

@inproceedings{rohatgi2020near,
  title={Near-Optimal Bounds for Online Caching with Machine Learned Advice},
  author={Rohatgi, Dhruv},
  booktitle={Proceedings of the 31st Annual ACM-SIAM Symposium on Discrete Algorithms (SODA)},
  pages={1834--1845},
  year={2020},
  eprint={1910.12172}
}

@inproceedings{chen2025guard,
  title={Robustifying Learning-Augmented Caching Efficiently without Compromising 1-Consistency},
  author={Chen, Peng and Zhao, Hailiang and Zhang, Jiaji and Tang, Xueyan and Wang, Yixuan and Deng, Shuiguang},
  booktitle={Advances in Neural Information Processing Systems (NeurIPS)},
  year={2025},
  url={https://neurips.cc/virtual/2025/loc/san-diego/poster/116615}
}

@article{thompson1933likelihood,
  title={On the Likelihood that One Unknown Probability Exceeds Another in View of the Evidence of Two Samples},
  author={Thompson, William R.},
  journal={Biometrika},
  volume={25},
  number={3/4},
  pages={285--294},
  year={1933},
  doi={10.1093/biomet/25.3-4.285}
}

@InProceedings{agrawal2012analysis,
  title = 	 {Analysis of Thompson Sampling for the Multi-armed Bandit Problem},
  author = 	 {Agrawal, Shipra and Goyal, Navin},
  booktitle = 	 {Proceedings of the 25th Annual Conference on Learning Theory},
  pages = 	 {39.1--39.26},
  year = 	 {2012},
  editor = 	 {Mannor, Shie and Srebro, Nathan and Williamson, Robert C.},
  volume = 	 {23},
  series = 	 {Proceedings of Machine Learning Research},
  address = 	 {Edinburgh, Scotland},
  month = 	 {25--27 Jun},
  publisher =    {PMLR},
  url = 	 {https://proceedings.mlr.press/v23/agrawal12.html}
}

@inproceedings{dar1996semantic,
  title={Semantic Data Caching and Replacement},
  author={Dar, Shaul and Franklin, Michael J. and J{\'o}nsson, Bj{\"o}rn {\TH}{\'o}r and Srivastava, Divesh and Tan, Michael},
  booktitle={Proceedings of the 22nd International Conference on Very Large Data Bases (VLDB)},
  pages={330--341},
  year={1996},
  url={https://www.vldb.org/conf/1996/P330.PDF}
}

@inproceedings{bang2023gptcache,
  title={{GPTCache}: An Open-Source Semantic Cache for {LLM} Applications},
  author={Bang, Fu},
  booktitle={Proceedings of the 3rd Workshop for Natural Language Processing Open Source Software (NLP-OSS)},
  pages={212--218},
  year={2023},
  doi={10.18653/v1/2023.nlposs-1.24}
}

@article{sun2025dam,
  title={Beyond Heuristics: A Decision-Theoretic Framework for Agent Memory Management},
  author={Sun, Changzhi and Chen, Xiangyu and Luo, Jixiang and Zhang, Dell and Li, Xuelong},
  journal={arXiv preprint arXiv:2512.21567},
  year={2025},
  eprint={2512.21567}
}

@article{zhang2026amac,
  title={Adaptive Memory Admission Control for {LLM} Agents},
  author={Zhang, Guilin and Jiang, Wei and Wang, Xiejiashan and Behr, Aisha and Zhao, Kai and Friedman, Jeffrey and Chu, Xu and Anoun, Amine},
  journal={arXiv preprint arXiv:2603.04549},
  year={2026},
  eprint={2603.04549}
}

@article{zhang2025memact,
  title={Memory as Action: Autonomous Context Curation for Long-Horizon Agentic Tasks},
  author={Zhang, Yuxiang and Shu, Jiangming and Ma, Ye and Lin, Xueyuan and Wu, Shangxi and Sang, Jitao},
  journal={arXiv preprint arXiv:2510.12635},
  year={2025},
  eprint={2510.12635}
}

@article{packer2023memgpt,
  title={{MemGPT}: Towards {LLMs} as Operating Systems},
  author={Packer, Charles and Fang, Vivian and Patil, Shishir G. and Lin, Kevin and Wooders, Sarah and Gonzalez, Joseph E.},
  journal={arXiv preprint arXiv:2310.08560},
  year={2023},
  eprint={2310.08560}
}

@article{wang2023voyager,
  title={Voyager: An Open-Ended Embodied Agent with Large Language Models},
  author={Wang, Guanzhi and Xie, Yuqi and Jiang, Yunfan and Mandlekar, Ajay and Xiao, Chaowei and Zhu, Yuke and Fan, Linxi and Anandkumar, Anima},
  journal={arXiv preprint arXiv:2305.16291},
  year={2023},
  eprint={2305.16291}
}

@inproceedings{park2023generative,
  title={Generative Agents: Interactive Simulacra of Human Behavior},
  author={Park, Joon Sung and O'Brien, Joseph C. and Cai, Carrie J. and Morris, Meredith Ringel and Liang, Percy and Bernstein, Michael S.},
  booktitle={Proceedings of the 36th Annual ACM Symposium on User Interface Software and Technology (UIST)},
  pages={1--22},
  year={2023},
  doi={10.1145/3586183.3606763}
}

@article{ai2025memorybench,
  title={{MemoryBench}: A Benchmark for Memory and Continual Learning in {LLM} Systems},
  author={Ai, Qingyao and Tang, Yichen and Wang, Changyue and Long, Jianming and Su, Weihang and Liu, Yiqun},
  journal={arXiv preprint arXiv:2510.17281},
  year={2025},
  eprint={2510.17281}
}

@article{jang2023dialsim,
  title={{DialSim}: A Real-Time Simulator for Evaluating Long-Term Dialogue Understanding of Conversational Agents},
  author={Kim, Jiho and Chay, Woosog and Hwang, Hyeonji and Kyung, Daeun and Chung, Hyunseung and Cho, Eunbyeol and Jo, Yohan and Choi, Edward},
  journal={arXiv preprint arXiv:2406.13144},
  year={2024},
  eprint={2406.13144}
}

@book{lattimore2020bandit,
  title={Bandit Algorithms},
  author={Lattimore, Tor and Szepesv{\'a}ri, Csaba},
  publisher={Cambridge University Press},
  year={2020},
  doi={10.1017/9781108571401}
}

@inproceedings{lewis2020retrieval,
  title={Retrieval-Augmented Generation for Knowledge-Intensive {NLP} Tasks},
  author={Lewis, Patrick and Perez, Ethan and Piktus, Aleksandra and Petroni, Fabio and Karpukhin, Vladimir and Goyal, Naman and K{\"u}ttler, Heinrich and Lewis, Mike and Yih, Wen-tau and Rockt{\"a}schel, Tim and Riedel, Sebastian and Kiela, Douwe},
  booktitle={Advances in Neural Information Processing Systems},
  volume={33},
  pages={9459--9474},
  year={2020},
  eprint={2005.11401}
}

@article{gao2024retrieval,
  title={Retrieval-Augmented Generation for Large Language Models: A Survey},
  author={Gao, Yunfan and Xiong, Yun and Gao, Xinyu and Jia, Kangxiang and Pan, Jinliu and Bi, Yuxi and Dai, Yi and Sun, Jiawei and Wang, Meng and Wang, Haofen},
  journal={arXiv preprint arXiv:2312.10997},
  year={2024},
  eprint={2312.10997}
}

@book{borodin2005online,
  title={Online Computation and Competitive Analysis},
  author={Borodin, Allan and El-Yaniv, Ran},
  publisher={Cambridge University Press},
  year={2005},
  isbn={9780521619462}
}

@inproceedings{o1993lru,
  title={The {LRU-K} Page Replacement Algorithm for Database Disk Buffering},
  author={O'Neil, Elizabeth J. and O'Neil, Patrick E. and Weikum, Gerhard},
  booktitle={Proceedings of the ACM SIGMOD International Conference on Management of Data},
  pages={297--306},
  year={1993},
  doi={10.1145/170035.170081}
}

@inproceedings{jiang2002lirs,
  title={{LIRS}: An Efficient Low Inter-reference Recency Set Replacement Policy to Improve Buffer Cache Performance},
  author={Jiang, Song and Zhang, Xiaodong},
  booktitle={Proceedings of the ACM SIGMETRICS International Conference on Measurement and Modeling of Computer Systems},
  pages={31--42},
  year={2002},
  doi={10.1145/511334.511340}
}

@inproceedings{song2020learning,
  title={Learning Relaxed {B}elady for Content Distribution Network Caching},
  author={Song, Zhenyu and Berger, Daniel S. and Li, Kai and Lloyd, Wyatt},
  booktitle={Proceedings of the 17th USENIX Symposium on Networked Systems Design and Implementation (NSDI)},
  pages={529--544},
  year={2020},
  url={https://www.usenix.org/conference/nsdi20/presentation/song}
}

@inproceedings{liu2020imitation,
  title={An Imitation Learning Approach for Cache Replacement},
  author={Liu, Evan Zheran and Hashemi, Milad and Swersky, Kevin and Ranganathan, Parthasarathy and Ahn, Junwhan},
  booktitle={Proceedings of the 37th International Conference on Machine Learning (ICML)},
  pages={6237--6247},
  year={2020},
  eprint={2006.16239}
}

@inproceedings{zhong2024memorybank,
  title={{MemoryBank}: Enhancing Large Language Models with Long-Term Memory},
  author={Zhong, Wanjun and Guo, Lianghong and Gao, Qiqi and Ye, He and Wang, Yanlin},
  booktitle={Proceedings of the AAAI Conference on Artificial Intelligence},
  volume={38},
  pages={19724--19731},
  year={2024},
  url={https://ojs.aaai.org/index.php/AAAI/article/view/29946}
}

@article{chhikara2025mem0,
  title={{Mem0}: Building Production-Ready {AI} Agents with Scalable Long-Term Memory},
  author={Chhikara, Prateek and others},
  journal={arXiv preprint arXiv:2504.19413},
  year={2025},
  eprint={2504.19413}
}

@inproceedings{locomo2024,
  title={Evaluating Very Long-Term Conversational Memory of {LLM} Agents},
  author={Maharana, Adyasha and Lee, Dong-Ho and Tuber, Sergey and Bansal, Mohit},
  booktitle={Proceedings of the 62nd Annual Meeting of the Association for Computational Linguistics (ACL)},
  year={2024},
  doi={10.18653/v1/2024.acl-long.747}
}

@article{zhang2024survey,
  title={A Survey on the Memory Mechanism of Large Language Model Based Agents},
  author={Zhang, Zeyu and others},
  journal={arXiv preprint arXiv:2404.13501},
  year={2024},
  eprint={2404.13501}
}

@misc{solar-techreport,
  title={{SOLAR}: Complete Proofs (Extended Supplement)},
  author={Sun, Yushi and Cao, Bowen and Lam, Wai},
  year={2026},
  howpublished={\url{https://github.com/ysunbp/SOLAR/blob/main/proof.pdf}},
  note={Accessed: 2026}
}

\end{document}